\documentclass[a4paper,USenglish,thm-restate]{lipics-v2021}

\usepackage{graphicx}
\usepackage{epsfig}
\usepackage{epsf}
\usepackage{latexsym}
\usepackage{bbm}
\usepackage{float}
\usepackage{fancyvrb}
\usepackage{url}
\usepackage{amsfonts}
\usepackage{amssymb}
\usepackage{amsmath}
\usepackage{enumerate, enumitem}

\usepackage{thmtools}

\usepackage{cite}
\usepackage{xcolor}
\usepackage{tikz}

\newcommand{\eps}{\varepsilon}

\newcommand{\conv}{\mathrm{conv}}
\newcommand{\diam}{\mathrm{diam}}

\newcommand{\DS}{\mathrm{DS}}

\newcommand{\NN}{\mathbb{N}}

\newcommand{\RR}{\mathbb{R}}

\def\P{\mathcal{P}}

\newcommand{\etal}{{et~al.}}
\newcommand{\ie}{{i.e.}}
\newcommand{\eg}{{e.g.}}

\newcommand{\later}[1]{{}}
\newcommand{\old}[1]{{}}
\long\def\ignore#1{}

\nolinenumbers 
\title{On the Stretch Factor of Polygonal Chains\thanks{A preliminary 
version of this paper appeared in the
\emph{Proceedings of the 44th International Symposium on Mathematical 
Foundations of Computer Science},
(MFCS 2019), Aachen, Germany, August 2019, Vol. 138 of LIPIcs, 56:1--56:14.}
}

\titlerunning{On the Stretch Factor of Polygonal Chains}

\author{Ke Chen}
{Department of Computer Science,
University of Wisconsin--Milwaukee, USA}
{kechen@uwm.edu}
{0000-0001-5470-6621}
{}
\author{Adrian Dumitrescu}
{Algoresearch L.L.C., Milwaukee, WI 53217, USA}
{ad.dumitrescu@gmail.com}
{0000-0002-1118-0321}%
{}
\author{Wolfgang Mulzer}
{Institut f\"ur Informatik,
Freie Universit\"at Berlin, Germany}
{mulzer@inf.fu-berlin.de}
{0000-0002-1948-5840}%
{Supported in part by ERC STG 757609.}
\author{Csaba D. T\'oth}
       {Department of Mathematics, California State University Northridge, Los Angeles, CA 91330-8313;
         and Department of Computer Science, Tufts University, Medford, MA 02155, USA}
{csaba.toth@csun.edu}
{0000-0002-8769-3190}
{Supported in part by NSF CCF-1422311, CCF-1423615, and DMS-1800734.}

\authorrunning{K.~Chen, A.~Dumitrescu, W.~Mulzer, and C.D.~T\'oth}

\Copyright{Ke Chen, Adrian Dumtrescu, Wolfgang Mulzer, Csaba D. T\'oth}



\acknowledgements{This work was initiated at the Fields Workshop on Discrete and Computational Geometry, held July 31--August 4, 2017, at Carleton University. The authors thank the organizers and all participants of the workshop for inspiring discussions and for providing a great research atmosphere. This problem was initially posed by Rolf Klein in 2005. We would like to thank Rolf Klein and Christian Knauer for interesting discussions on the stretch factor and related topics.}

\hideLIPIcs

\begin{document}

\maketitle

\begin{abstract}
Let $P=(p_1, p_2, \dots, p_n)$ be a polygonal chain in $\RR^d$.
The \emph{stretch factor} of $P$ is the ratio
between the total length of $P$ and the distance of its
endpoints, $\sum_{i = 1}^{n-1} |p_i p_{i+1}|/|p_1 p_n|$.
For a parameter $c \geq 1$, we call $P$ a \emph{$c$-chain} if
$|p_ip_j|+|p_jp_k| \leq c|p_ip_k|$, for every triple $(i,j,k)$, $1 \leq i<j<k \leq n$.
The stretch factor is a global property: it measures how close $P$ is
to a straight line, and it involves all the vertices of $P$; being a $c$-chain,
on the other hand, is a \emph{fingerprint}-property: it only depends on
subsets of $O(1)$ vertices of the chain.

We investigate how the $c$-chain property influences the stretch factor in the plane:
(i) we show that for every $\eps > 0$, there is a noncrossing $c$-chain that has
stretch factor $\Omega(n^{1/2-\eps})$, for sufficiently large constant $c=c(\eps)$;
(ii) on the other hand, the stretch factor of a $c$-chain $P$ is $O\left(n^{1/2}\right)$,
for every constant $c\geq 1$, regardless of whether $P$ is crossing or noncrossing; and
(iii) we give a randomized algorithm that can determine, for a polygonal chain
  $P$ in $\RR^2$ with $n$ vertices, the minimum $c\geq 1$ for which $P$ is a $c$-chain
  in $O\left(n^{2.5}\ {\rm polylog}\ n\right)$ expected time and $O(n\log n)$ space.
These results generalize to $\RR^d$. For every dimension $d\geq 2$ and every $\eps>0$,
we construct a noncrossing $c$-chain that has stretch factor
$\Omega\left(n^{(1-\eps)(d-1)/d}\right)$;
on the other hand, the stretch factor of any $c$-chain is $O\left((n-1)^{(d-1)/d}\right)$;
for every $c>1$, we can test whether an $n$-vertex chain in $\RR^d$ is a $c$-chain in
$O\left(n^{3-1/d}\ {\rm polylog}\ n\right)$ expected time and $O(n\log n)$ space.
\end{abstract}

\section{Introduction} \label{sec:intro}

Given a set $S$ of $n$ point sites in a Euclidean space $\RR^d$,
what is the best way to connect $S$ into a \emph{geometric network (graph)}?
This question has motivated researchers for a long time, going
back as far as the 1940s, and beyond~\cite{F40,Ve51}. Numerous
possible criteria for a good geometric network have been proposed,
perhaps the most basic being the \emph{length}.
In 1955, Few~\cite{Fe55} showed that for any set of $n$ points
in a unit square, there is a traveling salesman tour of length
at most $\sqrt{2n}+7/4$. This was improved to at most $0.984\sqrt{2n}+11$
by Karloff~\cite{Ka89}. Similar bounds hold for the shortest
spanning tree and the shortest rectilinear spanning
tree~\cite{CG81,DM18,GP68}. Besides length, two further key factors
in the quality of a geometric network are the \emph{vertex dilation}
and the \emph{geometric dilation}~\cite{NS07}, both of which measure
how closely shortest paths in a network approximate the
Euclidean distances between their endpoints.

The \emph{dilation} (also called \emph{stretch factor}~\cite{MM17}
or \emph{detour}~\cite{AKK+08})
between two points $p$ and $q$ in a geometric graph $G$ is defined as the ratio between
the length of a shortest path from $p$ to $q$ and the Euclidean distance $|pq|$.
The \emph{dilation} of the graph $G$ is the maximum dilation over all pairs of vertices in $G$.
A graph in which the dilation is bounded
above by $t \geq 1$ is also called a \emph{$t$-spanner}
(or simply a \emph{spanner} if $t$ is a constant).
A complete graph in Euclidean space is clearly a $1$-spanner. Therefore,
researchers focused on the dilation of graphs with certain additional constraints,
for example, noncrossing (\ie, plane) graphs. In 1989, Das and Joseph~\cite{DJ89}
identified a large class of plane spanners (characterized
by two simple local properties).
Bose~\etal\cite{BGS05} gave
an algorithm that constructs for any set of planar sites
a plane $11$-spanner with bounded degree.
On the other hand, Eppstein~\cite{Ep02}  analyzed
a fractal construction showing that \emph{$\beta$-skeletons}, a natural class
of geometric networks, can have arbitrarily large dilation.

The study of dilation also raises algorithmic questions.
Agarwal~\etal~\cite{AKK+08} described randomized algorithms for computing the dilation of
a given path (on $n$ vertices) in $\RR^2$ in $O(n \log n)$ expected time.
They also presented randomized algorithms for computing the dilation
of a given tree, or cycle, in $\RR^2$ in $O(n\log^2 n)$ expected time.
Previously, Narasimhan and Smid~\cite{NS00} showed that an $(1+\eps)$-approximation
of the stretch factor of any path, cycle, or tree can be computed in $O(n \log n)$ time.
Klein~\etal~\cite{KKNS09} gave randomized algorithms for a path, tree, or cycle in $\RR^2$
to count the number of vertex pairs whose dilation is below a given threshold in $O(n^{3/2+\eps})$
expected time.
Cheong~\etal~\cite{CHL08} showed that it is NP-hard to determine the existence of a
spanning tree on a planar point set whose dilation is at most a given value.
More results on plane spanners can be found in the monograph dedicated to this subject~\cite{NS07}
or in several surveys~\cite{Ep00,BS13,MM17}.

We investigate a basic question about the dilation of polygonal chains.
We ask how the dilation between the endpoints of a polygonal chain
(which we will call the \emph{stretch factor}, to distinguish it from
the more general notion of dilation) is influenced by \emph{fingerprint} properties of the
chain, \ie, by properties that are defined on $O(1)$-size subsets of the vertex set.
Such fingerprint properties play an important role in geometry; classic
examples include the \emph{Carath\'eodory property}\footnote{Given a finite set $S$ of points
  in $d$ dimensions, if every $d + 2$ points in $S$ are in convex position,
  then $S$ is in convex position.}~\cite[Theorem~1.2.3]{Mat02}
or the \emph{Helly property}\footnote{Given a finite collection of convex
  sets in $d$ dimensions, if every $d + 1$ sets have nonempty intersection,
  then all sets have nonempty intersection.}~\cite[Theorem~1.3.2]{Mat02}.
In general, determining the effect of a fingerprint property may prove elusive---given
$n$ points in the plane, consider the simple property that every $3$ points determine $3$
distinct distances. It is unknown~\cite[p.~203]{BMP05} whether this property implies
that the total number of distinct distances grows superlinearly in $n$.
Furthermore, fingerprint properties appear in the general study of
\emph{local versus global properties of metric spaces}, which is
highly relevant to combinatorial approximation algorithms based on
mathematical programming relaxations~\cite{ALNRRV12}.

In the study of dilation, interesting fingerprint properties
have also been found.
For example, a (continuous) curve $C$ is said to have the
\emph{increasing chord property}~\cite{CFG91, LM72} if for any points $a$, $b$, $c$, $d$
that appear on $C$ in this order, we have $|ad| \geq |bc|$.
The increasing chord property implies that $C$ has (geometric) dilation at most
$2\pi/3$~\cite{Ro94}. A weaker property is the \emph{self-approaching property}:
a (continuous) curve $C$ is self-approaching if for any points $a$, $b$, $c$ that appear
on $C$ in this order, we have $|ac| \geq |bc|$.
Self-approaching curves have dilation at most $5.332$~\cite{IKL99}
(see also~\cite{AAI+99}), and they have found interesting applications in the field of
graph drawing~\cite{ACG+12,BKL17,NPR16}.

We introduce a new natural fingerprint property and see that it can constrain the stretch factor
of a polygonal chain, but only in a weaker sense than one may expect; we also
provide algorithmic results on this property. Before providing details, we give
a few basic definitions.

\subparagraph*{Definitions.}
A \emph{polygonal chain} $P$ in $\RR^d$ is specified by a sequence of $n$ points
$(p_1, p_2, \dots, p_n)$, called \emph{vertices}.
The chain $P$ consists of $n-1$ line segments between consecutive vertices.
We say $P$ is \emph{simple} if only consecutive line segments intersect and they only
intersect at their endpoints.
Given a polygonal chain $P$ in $\RR^d$ with $n$ vertices and a parameter $c\geq 1$,
we call $P$ a \emph{$c$-chain} if for all $1\leq i < j < k\leq n$,
we have
\begin{equation} \label{eq:c-chain}
|p_ip_j|+|p_jp_k| \leq c|p_ip_k|.
\end{equation}
Observe that the $c$-chain condition is a fingerprint condition that is not really a local
dilation condition---it is more a combination between the local chain substructure
and the distribution of the points in the subchains.

The \emph{stretch factor} $\delta_P$ of $P$ is defined as the dilation between
the two end points $p_1$ and $p_n$ of the chain:
\[
\delta_P = \frac{\sum_{i = 1}^{n-1} |p_ip_{i+1}|}{|p_1p_n|}.
\]
Note that this definition is different from the more general notion of
dilation (also called \emph{stretch factor}~\cite{MM17}) of a graph
which is the maximum dilation over all pairs of vertices.
Since there is no ambiguity in this paper,
we will just call $\delta_P$ the stretch factor of $P$.

For example, the polygonal chain $P=((0,0), (1,0), \dots ,(n,0))$ in $\RR^2$ is a $1$-chain with
stretch factor $1$; and
$Q=((0,0), (0,1), (1,1), (1,0))$ is a $(\sqrt2+1)$-chain with stretch factor $3$.

Without affecting the results, the floor and ceiling functions are omitted in our calculations.
For a positive integer $t$, let $[t]=\{1,2,\dots,t\}$.
For a point set $S$, let $\conv(S)$ denote the convex hull of $S$.
All logarithms are in base 2, unless stated otherwise.

\subparagraph*{Our results.}
In the Euclidean plane $\RR^2$,
we deduce three upper bounds on the stretch factor of a $c$-chain $P$
with $n$ vertices (Section~\ref{sec:upper}).
In particular, we have
(i)~$\delta_P \leq c(n-1)^{\log c}$,
(ii)~$\delta_P\leq c(n-2)+1$, and
(iii)~$\delta_P =O\left(c^2\sqrt{n-1}\right)$.

From the other direction, we obtain the following lower bound in $\RR^2$ (Section~\ref{sec:lower}):
For every $c\geq 4$, there is a family $\P_c=\{P^m\}_{m\in\NN}$ of simple $c$-chains,
so that $P^m$ has $n=4^m+1$ vertices and stretch factor $(n-1)^{\frac{1+\log (c-2) - \log c}{2}}$,
where the exponent converges to $1/2$ as $c$ tends to infinity.
The lower bound construction does not extend to the case of $1<c<4$, which remains open.

Then we generalize the results to higher dimensional Euclidean spaces
(Section~\ref{sec:higher-dim}):
For all integers $d\geq 2$, we show that any $c$-chain $P$ with $n$ vertices in $\RR^d$
has stretch factor $\delta_P =O\left(c^2(n-1)^{(d-1)/d}\right)$.
On the other hand, for any constant $\eps>0$ and sufficiently large $c=\Omega(d)$,
we construct a $c$-chain in $\RR^d$ with $n$ vertices and stretch factor at least
$(n-1)^{(1-\eps)(d-1)/d}$.

Finally, we present two algorithmic results (Section~\ref{sec:algo}) for all fixed dimensions $d\geq 2$:
(i)~A randomized algorithm that decides, given a polygonal chain $P$ in $\RR^d$ with
$n$ vertices and a threshold $c>1$, whether $P$ is a
$c$-chain in $O\left(n^{3-1/d}\ {\rm polylog}\ n\right)$
expected time and $O(n\log n)$ space.
(ii)~As a corollary, there is a randomized algorithm that finds, for a polygonal chain
$P$ with $n$ vertices, the minimum $c\geq 1$ for which $P$ is a $c$-chain
in $O\left(n^{3-1/d}\ {\rm polylog}\ n\right)$ expected time and $O(n\log n)$ space.

\section{Upper Bounds in the Plane} \label{sec:upper}

At first glance, one might expect the stretch factor of a $c$-chain, for $c\geq 1$, to be bounded by
some function of $c$. For example, the stretch factor of a $1$-chain is necessarily $1$.
We derive three upper bounds on the stretch factor of a $c$-chain with $n$ vertices
in terms of $c$ and $n$ (cf.~Theorems~\ref{thm:logc}--\ref{thm:1/2});
see Fig.~\ref{fig:upperbds} for a visual comparison between the bounds.
For large $n$, the bound in Theorem~\ref{thm:logc} is the best for $1 \leq c  \leq 2^{1/2}$,
while the bound in Theorem~\ref{thm:1/2} is the best for $c > 2^{1/2}$.
In particular, the bound in Theorem~\ref{thm:logc} is tight for $c=1$.
When $n$ is comparable with $c$, more specifically,
for $c\geq 2$ and $n\leq 64c^2+2$,
the bound in Theorem~\ref{thm:linear} is the best.

\begin{figure}[!ht]
\centering
\includegraphics[width=0.5\textwidth]{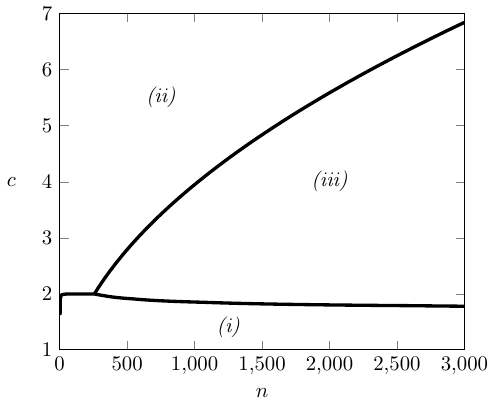}
\caption{The values of $n$ and $c$ for which (i) Theorem~\ref{thm:logc}: $\delta_P\leq c(n-1)^{\log c}$,
(ii) Theorem~\ref{thm:linear}: $\delta_P\leq c(n-2)+1$,
  and (iii) Theorem~\ref{thm:1/2}: $\delta_P\leq 8c^2\sqrt{n-1}$
  give the current best upper bound.\label{fig:upperbds}}
\end{figure}

Our first upper bound is obtained by a recursive application of the
$c$-chain property. It holds for any positive distance
function that
need not even
satisfy the triangle inequality.
\begin{theorem}\label{thm:logc}
For a $c$-chain $P$ with $n$ vertices, we have $\delta_P \leq c(n-1)^{\log c}$.
\end{theorem}
\begin{proof}
We prove, by induction on $n$, that
\begin{equation}\label{eq:logc}
\delta_P\leq c^{\left\lceil\log (n-1)\right\rceil},
\end{equation}
for every $c$-chain $P$ with $n \geq 2$ vertices. In the base case, $n=2$, we
have $\delta_P=1$ and $c^{\left\lceil\log (2-1)\right\rceil}=1$.
Now let $n \geq 3$, and assume that (\ref{eq:logc}) holds for every $c$-chain
with fewer than $n$ vertices.
Let $P = (p_1, \dots, p_n)$ be a $c$-chain with $n$ vertices. Then,
applying (\ref{eq:logc}) to the first and second half of $P$, followed
by the $c$-chain property for the first, middle, and last vertex of $P$, we get
\begin{align*}
\sum_{i=1}^{n-1}|p_{i}p_{i+1}|
&\leq \sum_{i=1}^{\lceil n/2\rceil-1}|p_{i}p_{i+1}| + \sum_{i=\lceil n/2\rceil}^{n-1} |p_{i}p_{i+1}|\\
&\leq c^{\left\lceil\log (\lceil n/2\rceil-1) \right\rceil} \left( |p_1p_{\lceil n/2\rceil}| + |p_{\lceil n/2\rceil}p_n|\right)\\
&\leq c^{\left\lceil\log (\lceil n/2\rceil-1) \right\rceil}\cdot  c|p_1p_n|\\
&\leq  c^{\left\lceil\log (n-1)\right\rceil} |p_1p_n|,
\end{align*}
so (\ref{eq:logc}) holds also for $P$.
Consequently,
\[ \delta_P \leq c^{\left\lceil\log (n-1)\right\rceil} \leq c^{\log (n-1) +1} =c \cdot c^{\log (n-1)}
= c \, (n-1)^{\log{c}}, \]
as required.
\end{proof}

Our second upper bound combines the $c$-chain property with the triangle
inequality, and it holds in any metric space.
\begin{theorem}\label{thm:linear}
For a $c$-chain $P$ with $n$ vertices, we have $\delta_P\leq c(n-2)+1$.
\end{theorem}
\begin{proof}
Without loss of generality, assume that $|p_1p_n|=1$.
For every $1<i<n$, the $c$-chain property implies
$|p_1p_i|+|p_ip_n|\leq c|p_1p_n| = c$, hence
\begin{equation}\label{eq:metric1}
|p_1p_i|\leq c-|p_ip_n|.
\end{equation}
The triangle inequality yields
\begin{equation}\label{eq:metric2}
|p_1p_i|\leq |p_1p_n|+|p_np_i|=1+|p_ip_n|.
\end{equation}
The combination of \eqref{eq:metric1} and \eqref{eq:metric2} gives
$|p_1p_i|\leq \frac{c+1}{2}$. Analogous argument for $p_n$ (in place of $p_1$)
yields $|p_ip_n|\leq \frac{c+1}{2}$.

For every pair $1< i<j<n$, the triangle inequality implies
\begin{equation*}
2|p_ip_j|
\leq (|p_ip_1|+|p_1p_j|)+(|p_ip_n|+|p_np_j|)
=   (|p_1p_i|+|p_ip_n|)+(|p_1p_j|+|p_jp_n|)\leq 2c,
\end{equation*}
hence $|p_ip_j|\leq c$.
Overall, the stretch factor of $P$ is bounded above by
\begin{align*}
\delta_P &= \frac{\sum_{j=1}^{n-1}|p_{j}p_{j+1}|}{|p_1p_n|}
= |p_1p_2|+|p_{n-1}p_n|+\sum_{j=2}^{n-2}|p_jp_{j+1}| \\
&\leq \frac{c+1}{2}+\frac{c+1}{2}+c(n-3) = c(n-2)+1,
\end{align*}
as claimed.
\end{proof}

Our third upper bound uses properties of the Euclidean plane
(specifically, a volume argument) to bound the number of long edges in $P$.
\begin{theorem}\label{thm:1/2}
For a $c$-chain $P$ with $n$ vertices,
we have $\delta_P =O\left(c^2\sqrt{n-1}\right)$.
\end{theorem}
\begin{proof}
Let $P=(p_1,\dots, p_n)$ be a $c$-chain, for some constant $c \geq 1$,
and let $L=\sum_{i=1}^{n-1}|p_ip_{i+1}|$ be its length.
We may assume that $p_1 p_n$ is a horizontal segment of unit length.
By the $c$-chain property, every point $p_j$, $1 < j < n$, lies in an ellipse $E$
with foci $p_1$ and $p_n$;
see \textsc{Fig.}~\ref{fig:ellipse}.
The diameter of $E$ is its major axis, whose length is $c$.
Let $U$ be a disk of radius $c/2$ concentric with $E$, and note that $E\subset U$
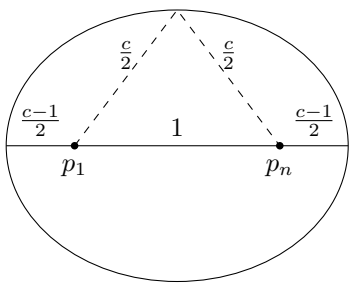
\begin{figure}[htpb]
\centering
\begin{tikzpicture}[scale=0.5]
\draw (5,0) circle [x radius=5, y radius=4];
\draw (0,0) --node[above]{$\frac{c-1}{2}$} (2,0)node[circle, fill, inner sep=1pt, label=below:$p_1$]{}
-- node[above]{$1$} (8,0)node[circle, fill, inner sep=1pt, label=below:$p_n$]{}
-- node[above]{$\frac{c-1}{2}$} (10,0);
\draw[dashed] (2,0) -- node[left=2pt]{$\frac{c}{2}$} (5, 4)
-- node[right=2pt]{$\frac{c}{2}$} (8, 0);
\draw (5,0) circle [radius=5];
\draw[dashed] (5,0)node[circle, fill, inner sep=1pt]{} -- node[left]{$\frac{c}{2}$} (5,-5);
\node at (7, -3) {$E$};
\node at (7, -4.2) {$U$};
\end{tikzpicture}
\caption{The entire chain $P$ lies in an ellipse $E$ with foci $p_1$ and $p_n$.
$E$ lies in a cocentric disk $U$ of radius $c/2$.}\label{fig:ellipse}
\end{figure}

We set $x=4c^2/\sqrt{n-1}$; and let $L_0$ and $L_1$ be the sum of lengths of all edges in $P$
of length at most $x$ and more than $x$, respectively. By definition, we have $L=L_0+L_1$ and
\begin{equation}\label{eq:short}
L_0\leq (n-1)x = (n-1)\cdot 4c^2/\sqrt{n-1} = 4c^2\sqrt{n-1}.
\end{equation}
We shall prove that $L_1\leq 4c^2\sqrt{n-1}$, implying $L\leq 8c^2\sqrt{n-1}$.
For this, we further classify the edges in $L_1$ according to their lengths:
For $\ell=0,1,\dots , \infty$, let
\begin{equation}\label{eq:PL1}
P_\ell=\left\{p_i: 2^\ell x< |p_ip_{i+1}|\leq 2^{\ell+1}x\right\}.
\end{equation}
Since all points lie in an ellipse of diameter $c$, we have $|p_ip_{i+1}|\leq c$, for all $i=0,\dots, n-1$.
Consequently, $P_\ell=\emptyset$ when $c\leq 2^\ell x$, or equivalently $\log(c/x)\leq \ell$.

We use a volume argument to derive an upper bound on the cardinality of $P_\ell$, for
$\ell=0,1,\dots , \lfloor\log (c/x)\rfloor$.
Assume that $p_i,p_k\in P_\ell$, and w.l.o.g., $i<k$. If $k=i+1$, then by~\eqref{eq:PL1}, $2^\ell x<|p_ip_k|$.
Otherwise,
\[
2^\ell x<|p_ip_{i+1}| < |p_ip_{i+1}|+|p_{i+1}p_k| \leq c|p_ip_k|,
\text{  or  } \frac{2^\ell x}{c} < |p_ip_k|.
\]
Consequently, the disks of radius
\begin{equation}\label{area:Rad}
R=\frac{2^\ell x}{2c} = \frac{2\cdot 2^\ell c }{\sqrt{n-1}}
\end{equation}
centered at the points in $P_{\ell}$ are interior-disjoint.
The area of each disk is $\pi R^2$. Since $P_\ell\subset U$, these disks are contained
in the $R$-neighborhood $U_R$ of the disk $U$, which is a disk of radius $\frac{c}{2}+R$
concentric with $U$.
For $\ell \leq \log(c/x)$, we have $2^\ell x\leq c$, hence $R=\frac{2^\ell x}{2c} \leq \frac{c}{2c}=\frac12 \leq \frac{c}{2}$. Thus the radius of $U_R$ is at most $c$.
Since $U_R$ contains $|P_\ell|$ interior-disjoint disks of radius $R$, we obtain
\begin{equation}\label{eq:PL2}
|P_\ell|
\leq \frac{{\rm area}(U_R)}{\pi R^2}
< \frac{\pi c^2}{\pi R^2}
=\frac{4c^4}{2^{2\ell} x^2}.
\end{equation}
For every segment $p_{i-1}p_i$ with length more than $x$, we have that $p_i\in P_\ell$,
for some $\ell\in \{0,1,\dots , \lfloor\log (c/x)\rfloor\}$.
The total length of these segments is
\begin{align*}
L_1&\leq  \sum_{\ell=0}^{\lfloor\log (c/x)\rfloor} |P_\ell| \cdot 2^{\ell+1}x
 <   \sum_{\ell=0}^{\lfloor\log (c/x)\rfloor} \frac{4c^4}{2^{2\ell} x^2} \cdot 2^{\ell+1}x
 =  \sum_{\ell=0}^{\lfloor\log (c/x)\rfloor} \frac{8 c^4}{2^\ell x}\\
&<   \frac{8 c^4}{x} \sum_{\ell=0}^\infty \frac{1}{2^\ell}
 =  \frac{16 c^4}{x}
 =   4 c^2\cdot \sqrt{n-1},
\end{align*}
as required. Together with \eqref{eq:short}, this yields $
L\leq 8c^2\cdot \sqrt{n-1}$.
\end{proof}

\section{Lower Bounds in the Plane} \label{sec:lower}

We now present our lower bound construction, showing that the dependence on $n$
for the stretch factor of a $c$-chain cannot be avoided.

\begin{theorem}\label{thm:lower-bound}
For every constant $c\geq 4$, there is a set $\P_c=\{P^m\}_{m\in\NN}$ of simple $c$-chains,
so that $P^m$ has $n=4^m+1$ vertices and stretch factor $(n-1)^{\frac{1+\log (c-2) - \log c}{2}}$.
\end{theorem}

By Theorem~\ref{thm:1/2}, the stretch factor of a $c$-chain in the plane
is $O\left((n-1)^{1/2}\right)$ for every constant $c\geq 1$.
Since
\[ \lim_{c \to\infty} \frac{1+\log (c-2) - \log c}{2} =\frac12, \]
our lower bound construction shows that the limit of the exponent cannot be improved.
Indeed, for every $\eps>0$, we can set $c=\frac{2^{2\eps+1}}{2^{2\eps}-1}$,
and then the chains above have stretch factor
\[ (n-1)^{\frac{1+\log(c-2)-\log c}{2}}=(n-1)^{1/2-\eps}=\Omega(n^{1/2-\eps}). \]

We first construct a family $\P_c=\{P^m\}_{m\in\NN}$ of polygonal chains.
Then we show, in Lemmata~\ref{lemma:simple} and~\ref{lemma:c-chain},
that every chain in $\P_c$ is simple and indeed a $c$-chain.
The theorem follows since the claimed stretch factor is a consequence of the construction.

\subparagraph*{Construction of $\P_c$.}
The construction here is a generalization of the iterative construction of the \emph{Koch curve};
when $c=6$, the result is the original Ces\`aro fractal (which is a variant of the Koch curve)~\cite{Ces05}.
We start with a unit line segment $P^0$, and for $m=0, 1, \dots$,
we construct $P^{m+1}$ by replacing each segment in $P^m$ by four segments such that
the middle three points achieve a stretch factor of $c_*=\frac{c-2}{2}$ (this choice will be justified
in the proof of Lemma~\ref{lemma:c-chain}). Note that $c_*\geq 1$, since $c\geq 4$.

We continue with the details. Let $P^0$ be the unit line segment from $(0,0)$ to $(1,0)$;
see \textsc{Fig.}~\ref{fig:p0p1}\,(left).
Given the polygonal chain $P^m$ $(m=0,1,\dots$), we construct $P^{m+1}$ by replacing
each segment of $P^m$ by four segments as follows. Consider a segment of $P^m$,
and denote its length by $\ell$. Subdivide this segment into three segments of
lengths $(\frac{1}{2}-\frac{a}{c_*})\ell$, $\frac{2a}{c_*}\ell$, and $(\frac{1}{2}-\frac{a}{c_*})\ell$,
respectively, where $0<a<\frac{c_*}{2}$ is a parameter to be determined later.
Replace the middle segment with the top part of an isosceles triangle of side length $a\ell$.
The chains $P^0$, $P^1$, $P^2$, and $P^4$ are depicted in Figures~\ref{fig:p0p1} and~\ref{fig:p2p4}.

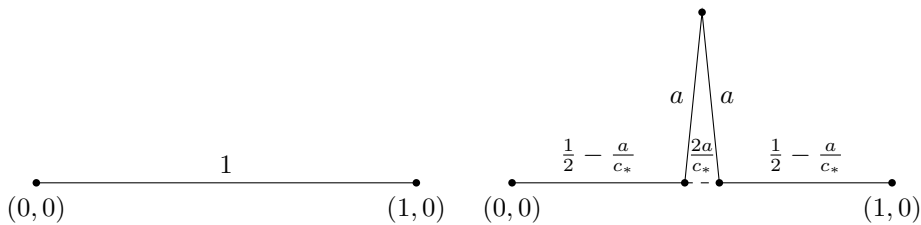
\begin{figure}[htbp!]
\centering
\begin{tikzpicture}[scale=0.5]
\draw (0.000, 0.000)node[circle, fill, inner sep=1pt, label=below:{$(0,0)$}]{} --
node[above]{1} (10.000, 0.000)node[circle, fill, inner sep=1pt, label=below:{$(1,0)$}]{};
\end{tikzpicture}
\hfill
\centering
\begin{tikzpicture}[scale=0.5]
\draw (0.000, 0.000)node[circle, fill, inner sep=1pt, label=below:{$(0,0)$}]{}
-- node[above]{$\frac{1}{2}-\frac{a}{c_*}$}
(4.545, 0.000) node[circle, fill, inner sep=1pt]{}
-- node[left]{$a$}(5.000, 4.523) node[circle, fill, inner sep=1pt]{}
-- node[right]{$a$}
(5.455, 0.000) node[circle, fill, inner sep=1pt]{}
-- node[above]{$\frac{1}{2}-\frac{a}{c_*}$}
(10.000, 0.000)node[circle, fill, inner sep=1pt, label=below:{$(1,0)$}]{} ;
\draw[dashed] (4.545, 0.000) -- node[above]{$\frac{2a}{c_*}$} (5.455, 0.000);
\end{tikzpicture}
\caption{The chains $P^0$ (left)  and $P^1$ (right).}\label{fig:p0p1}
\end{figure}

Note that each segment of length $\ell$ in $P^m$ is replaced by four segments of total
length $(1+\frac{2a(c_*-1)}{c_*})\ell$.
After $m$ iterations, the chain $P^m$ consists of $4^m$ line segments of total
length $\left(1+\frac{2a(c_*-1)}{c_*}\right)^m$.

By construction, the chain $P^m$ (for $m\geq 1$) consists of four scaled copies of $P^{m-1}$.
For $i=1,2,3,4$, let the \emph{$i$th subchain of $P^m$} be the
subchain of $P^m$ consisting of $4^{m-1}$ segments starting from the $((i-1)4^{m-1}+1)$th segment.
By construction, the $i$th subchain of $P^m$ is similar to the chain $P^{m-1}$,
for $i=1,2,3,4$.\footnote{Two geometric shapes are \emph{similar} if one can be obtained from
  the other by translation, rotation, and scaling; and are \emph{congruent} if one can be obtained
  from the other by translation and rotation.}
The following functions allow us to refer to these subchains formally.
For $i=1,2,3,4$, define a function $f^m_i: P^m\to P^m$ as the identity on the $i$th subchain of $P^m$ that
sends the remaining part(s) of $P^m$ to the closest endpoint(s) along this subchain.
So $f^m_i(P^m)$ is similar to $P^{m-1}$. Let $g_i: \P_c\setminus\{P^0\} \to\P_c$ be a piecewise defined
function such that $g_i(C)=\sigma^{-1}\circ f^m_i\circ\sigma(C)$ if $C$ is similar to $P^m$,
where $\sigma: C\to P^m$ is a similarity transformation.
Applying the function $g_i$ on a chain $P^m$ can be thought of as ``cutting out''
its $i$th subchain.

\tikzset{p4/.pic={code={
\draw (0.000, 0.000) -- (0.427, 0.000) -- (0.470, 0.425) -- (0.512,
0.000) -- (0.939, 0.000) -- (0.982, 0.425) -- (0.563, 0.510) --
(0.990, 0.510) -- (1.033, 0.934) -- (1.076, 0.510) -- (1.503, 0.510)
-- (1.084, 0.425) -- (1.127, 0.000) -- (1.554, 0.000) -- (1.597,
0.425) -- (1.639, 0.000) -- (2.066, 0.000) -- (2.109, 0.425) --
(1.690, 0.510) -- (2.117, 0.510) -- (2.160, 0.934) -- (1.742, 1.019)
-- (1.615, 0.612) -- (1.658, 1.036) -- (1.240, 1.121) -- (1.667,
1.121) -- (1.709, 1.546) -- (1.752, 1.121) -- (2.179, 1.121) --
(2.222, 1.546) -- (1.803, 1.631) -- (2.230, 1.631) -- (2.273, 2.056)
-- (2.315, 1.631) -- (2.742, 1.631) -- (2.324, 1.546) -- (2.367,
1.121) -- (2.794, 1.121) -- (2.836, 1.546) -- (2.879, 1.121) --
(3.306, 1.121) -- (2.887, 1.036) -- (2.930, 0.612) -- (2.804, 1.019)
-- (2.385, 0.934) -- (2.428, 0.510) -- (2.855, 0.510) -- (2.437,
0.425) -- (2.479, 0.000) -- (2.906, 0.000) -- (2.949, 0.425) --
(2.992, 0.000) -- (3.418, 0.000) -- (3.461, 0.425) -- (3.043, 0.510)
-- (3.470, 0.510) -- (3.512, 0.934) -- (3.555, 0.510) -- (3.982,
0.510) -- (3.564, 0.425) -- (3.606, 0.000) -- (4.033, 0.000) --
(4.076, 0.425) -- (4.119, 0.000) -- (4.545, 0.000) -- (4.588, 0.425)
-- (4.170, 0.510) -- (4.597, 0.510) -- (4.639, 0.934) -- (4.221,
1.019) -- (4.095, 0.612) -- (4.137, 1.036) -- (3.719, 1.121) --
(4.146, 1.121) -- (4.189, 1.546) -- (4.231, 1.121) -- (4.658, 1.121)
-- (4.701, 1.546) -- (4.282, 1.631) -- (4.709, 1.631) -- (4.752,
2.056) -- (4.334, 2.141) -- (4.207, 1.733) -- (4.250, 2.158) --
(3.832, 2.243) -- (3.705, 1.835) -- (4.098, 1.668) -- (3.680, 1.753)
-- (3.554, 1.346) -- (3.596, 1.770) -- (3.178, 1.855) -- (3.605,
1.855) -- (3.648, 2.280) -- (3.229, 2.365) -- (3.103, 1.957) --
(3.146, 2.382) -- (2.727, 2.467) -- (3.154, 2.467) -- (3.197, 2.892)
-- (3.240, 2.467) -- (3.666, 2.467) -- (3.709, 2.892) -- (3.291,
2.977) -- (3.718, 2.977) -- (3.760, 3.401) -- (3.803, 2.977) --
(4.230, 2.977) -- (3.812, 2.892) -- (3.854, 2.467) -- (4.281, 2.467)
-- (4.324, 2.892) -- (4.367, 2.467) -- (4.793, 2.467) -- (4.836,
2.892) -- (4.418, 2.977) -- (4.845, 2.977) -- (4.887, 3.401) --
(4.469, 3.486) -- (4.343, 3.079) -- (4.385, 3.503) -- (3.967, 3.588)
-- (4.394, 3.588) -- (4.437, 4.013) -- (4.479, 3.588) -- (4.906,
3.588) -- (4.949, 4.013) -- (4.530, 4.098) -- (4.957, 4.098) --
(5.000, 4.523) -- (5.043, 4.098) -- (5.470, 4.098) -- (5.051, 4.013)
-- (5.094, 3.588) -- (5.521, 3.588) -- (5.563, 4.013) -- (5.606,
3.588) -- (6.033, 3.588) -- (5.615, 3.503) -- (5.657, 3.079) --
(5.531, 3.486) -- (5.113, 3.401) -- (5.155, 2.977) -- (5.582, 2.977)
-- (5.164, 2.892) -- (5.207, 2.467) -- (5.633, 2.467) -- (5.676,
2.892) -- (5.719, 2.467) -- (6.146, 2.467) -- (6.188, 2.892) --
(5.770, 2.977) -- (6.197, 2.977) -- (6.240, 3.401) -- (6.282, 2.977)
-- (6.709, 2.977) -- (6.291, 2.892) -- (6.334, 2.467) -- (6.760,
2.467) -- (6.803, 2.892) -- (6.846, 2.467) -- (7.273, 2.467) --
(6.854, 2.382) -- (6.897, 1.957) -- (6.771, 2.365) -- (6.352, 2.280)
-- (6.395, 1.855) -- (6.822, 1.855) -- (6.404, 1.770) -- (6.446,
1.346) -- (6.320, 1.753) -- (5.902, 1.668) -- (6.295, 1.835) --
(6.168, 2.243) -- (5.750, 2.158) -- (5.793, 1.733) -- (5.666, 2.141)
-- (5.248, 2.056) -- (5.291, 1.631) -- (5.718, 1.631) -- (5.299,
1.546) -- (5.342, 1.121) -- (5.769, 1.121) -- (5.811, 1.546) --
(5.854, 1.121) -- (6.281, 1.121) -- (5.863, 1.036) -- (5.905, 0.612)
-- (5.779, 1.019) -- (5.361, 0.934) -- (5.403, 0.510) -- (5.830,
0.510) -- (5.412, 0.425) -- (5.455, 0.000) -- (5.881, 0.000) --
(5.924, 0.425) -- (5.967, 0.000) -- (6.394, 0.000) -- (6.436, 0.425)
-- (6.018, 0.510) -- (6.445, 0.510) -- (6.488, 0.934) -- (6.530,
0.510) -- (6.957, 0.510) -- (6.539, 0.425) -- (6.582, 0.000) --
(7.008, 0.000) -- (7.051, 0.425) -- (7.094, 0.000) -- (7.521, 0.000)
-- (7.563, 0.425) -- (7.145, 0.510) -- (7.572, 0.510) -- (7.615,
0.934) -- (7.196, 1.019) -- (7.070, 0.612) -- (7.113, 1.036) --
(6.694, 1.121) -- (7.121, 1.121) -- (7.164, 1.546) -- (7.206, 1.121)
-- (7.633, 1.121) -- (7.676, 1.546) -- (7.258, 1.631) -- (7.685,
1.631) -- (7.727, 2.056) -- (7.770, 1.631) -- (8.197, 1.631) --
(7.778, 1.546) -- (7.821, 1.121) -- (8.248, 1.121) -- (8.291, 1.546)
-- (8.333, 1.121) -- (8.760, 1.121) -- (8.342, 1.036) -- (8.385,
0.612) -- (8.258, 1.019) -- (7.840, 0.934) -- (7.883, 0.510) --
(8.310, 0.510) -- (7.891, 0.425) -- (7.934, 0.000) -- (8.361, 0.000)
-- (8.403, 0.425) -- (8.446, 0.000) -- (8.873, 0.000) -- (8.916,
0.425) -- (8.497, 0.510) -- (8.924, 0.510) -- (8.967, 0.934) --
(9.010, 0.510) -- (9.437, 0.510) -- (9.018, 0.425) -- (9.061, 0.000)
-- (9.488, 0.000) -- (9.530, 0.425) -- (9.573, 0.000) -- (10.000,
0.000);
}}}

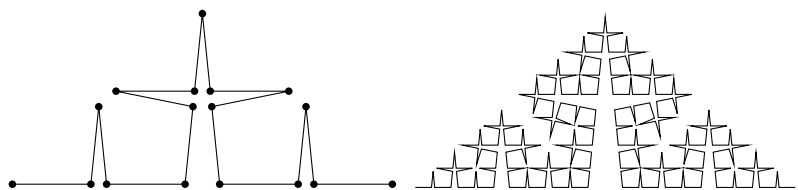
\begin{figure}[!ht]

\centering
\begin{tikzpicture}[scale=0.6]
\draw (0.000, 0.000)node[circle, fill, inner sep=1pt]{}
 -- (2.066, 0.000)node[circle, fill, inner sep=1pt]{}
 -- (2.273, 2.056)node[circle, fill, inner sep=1pt]{}
 -- (2.479, 0.000)node[circle, fill, inner sep=1pt]{}
 -- (4.545, 0.000)node[circle, fill, inner sep=1pt]{}
 -- (4.752, 2.056)node[circle, fill, inner sep=1pt]{}
 -- (2.727, 2.467)node[circle, fill, inner sep=1pt]{}
 -- (4.793, 2.467)node[circle, fill, inner sep=1pt]{}
 -- (5.000, 4.523)node[circle, fill, inner sep=1pt]{}
 -- (5.207, 2.467)node[circle, fill, inner sep=1pt]{}
 -- (7.273, 2.467)node[circle, fill, inner sep=1pt]{}
 -- (5.248, 2.056)node[circle, fill, inner sep=1pt]{}
 -- (5.455, 0.000)node[circle, fill, inner sep=1pt]{}
 -- (7.521, 0.000)node[circle, fill, inner sep=1pt]{}
 -- (7.727, 2.056)node[circle, fill, inner sep=1pt]{}
 -- (7.934, 0.000)node[circle, fill, inner sep=1pt]{}
 -- (10.000, 0.000)node[circle, fill, inner sep=1pt]{};
\end{tikzpicture}
\hfill
\begin{tikzpicture}
\path (0,0) pic[scale=0.6] {p4};
\end{tikzpicture}
\caption{The chains $P^2$ (left) and $P^4$ (right).}\label{fig:p2p4}
\end{figure}

Clearly, the stretch factor of the chain monotonically increases with the parameter $a$.
However, if $a$ is too large, the chain is no longer simple.
The following lemma gives a sufficient condition for the constructed chains to avoid self-crossings.

\begin{lemma}\label{lemma:simple}
For every constant $c\geq 4$, if $a\leq \frac{c-2}{2c}$, then every chain in $\P_c$ is simple.
\end{lemma}
\begin{proof}
  Let $T=\conv(P^1)$. Observe that $T$ is an isosceles triangle;
  see \textsc{Fig.}~\ref{fig:bounding-triangle}\,(left). We first show the following:

\begin{claim*}If $a\leq \frac{c-2}{2c}$, then $\conv(P^m) = T$
    for all $m \geq 1$.
\end{claim*}

\begin{proof}
We prove the claim by induction on $m$. It holds for $m=1$ by definition.
For the induction step, assume that $m\geq 2$ and that the claim holds for $m-1$.
Consider the chain $P^{m}$. Since it contains all the vertices of $P^1$, $T\subset \conv(P^{m})$.
So we only need to show that $\conv(P^{m})\subset T$.

\begin{figure}[htbp!]
\centering
\begin{tikzpicture}[scale=0.6]
\fill[black!15] (0.000, 0.000) -- (5.000, 4.623) -- (10.000, 0.000) --
cycle;
\draw (0.000, 0.000) -- node[above]{$\frac{1}{2}-\frac{a}{c_*}$} (4.545, 0.000)
-- node[left]{$a$} (5.000, 4.523)
-- node[right]{$a$} (5.455, 0.000) -- node[above]{$\frac{1}{2}-\frac{a}{c_*}$} (10.000, 0.000);
\draw[dashed] (4.545, 0.000) -- node[above]{$\frac{2a}{c_*}$} (5.455, 0.000);
\end{tikzpicture}
\hfill
\begin{tikzpicture}[scale=0.6]
\fill[black!15] (0.000, 0.000) -- (5.000, 4.623) -- (10.000, 0.000) -- cycle;
\fill[black!25] (0.000, 0.000) -- (2.273, 2.056) -- (4.545, 0.000) --cycle;
\fill[black!25] (4.545, 0.000) -- (2.727, 2.467) -- (5.000, 4.623) --cycle;
\fill[black!25] (5.000, 4.623) -- (7.273, 2.467) -- (5.455, 0.000) --cycle;
\fill[black!25] (5.455, 0.000) -- (7.727, 2.056) -- (10.000, 0.000) --cycle;
\draw (0.000, 0.000) -- (2.066, 0.000)
-- (2.273, 2.056) -- (2.479, 0.000) -- (4.545, 0.000) -- (4.752, 2.056)
-- (2.727, 2.467)node[circle, fill, inner sep=1pt, label=above:$p$]{}
-- (4.793, 2.467)
-- (5.000, 4.523)node[circle, fill, inner sep=1pt, label=above:$t$]{}
-- (5.207, 2.467) -- (7.273, 2.467) -- (5.248, 2.056) -- (5.455, 0.000)
-- (7.521, 0.000) -- (7.727, 2.056) -- (7.934, 0.000) -- (10.000, 0.000);
\node (label1) at (1.5, 4) {\small $a\left(\frac{1}{2}-\frac{a}{c_*}\right)$};
\draw [-latex] (label1.south) to [bend right=20] (2.211, 1.439);
\node (label2) at (0, 2) {\small $\left(\frac{1}{2}-\frac{a}{c_*}\right)^2$};
\draw [-latex] (label2.south) to [out=270, in=90] (1.1, 0.1);
\end{tikzpicture}
\caption{Left: Convex hull $T$ of $P^1$ in light gray;
 Right: Convex hulls of $g_i(P^2)$, $i=1,2,3,4$, in dark gray,
 are contained in $T$.}
\label{fig:bounding-triangle}
\end{figure}
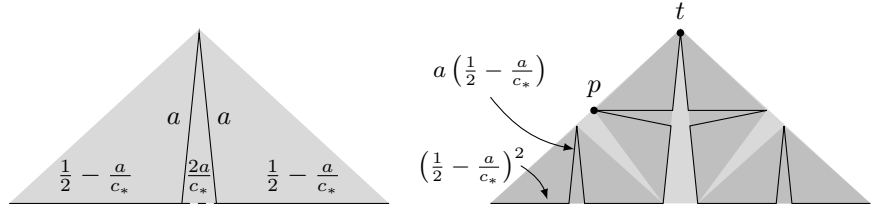

By construction, $P^{m}\subset \bigcup_{i=1}^4 \conv(g_i(P^{m}))$;
see \textsc{Fig.}~\ref{fig:bounding-triangle}\,(right).
By the inductive hypothesis, $\conv(g_i(P^{m}))$ is an isosceles triangle similar to $T$, for $i=1,2,3,4$.
Since the bases of $\conv(g_1(P^{m}))$ and $\conv(g_4(P^{m}))$ are
collinear with the base of $T$ by construction, due to similarity, they are contained in $T$.
The base of $\conv(g_2(P^{m}))$ is contained in $T$. In order to show
$\conv(g_2(P^{m}))\subset T$, by convexity, it suffices to ensure that its apex $p$ is also in $T$.
Note that the coordinates of the top point are $t=\left(1/2, a\sqrt{c_*^2-1}/c_* \right)$,
so the supporting line $\ell$ of the left side of $T$ is
\begin{align*}
y &=\frac{2a\sqrt{c_*^2-1}}{c_*}x, \text{ and} \\
p &=\left(\frac{1}{2}-\frac{a}{2c_*}-\frac{a^2\left(c_*^2-1\right)}{c_*^2},
\left(\frac{a}{2c_*}+\frac{a^2}{c_*^2}\right)\sqrt{c_*^2-1}\right).
\end{align*}
By the condition of $a\leq \frac{c-2}{2c}=\frac{c_*}{2(c_*+1)}$ in the lemma,
$p$ lies on or below $\ell$.
Under the same condition, we have $\conv(g_3(P^{m})) \subset T$ by symmetry.
Then $P^{m}\subset \bigcup_{i=1}^4 \conv(g_i(P^{m}))\subset T$.
Since $T$ is convex, $\conv(P^m)\subset T$.
So $\conv(P^m) = T$, as claimed.
\end{proof}

We can now finish the proof of Lemma~\ref{lemma:simple} by induction.
Clearly, $P^0$ and $P^1$ are simple. Assume that $m \geq 2$, and that $P^{m-1}$ is simple.
Consider the chain $P^{m}$. For $i=1,2,3,4$, $g_i(P^{m})$ is similar to $P^{m-1}$,
and hence simple by the inductive hypothesis.
Since $P^{m}=\bigcup_{i=1}^4 g_i(P^{m})$, it is sufficient to show that for all $i,j\in\{1,2,3,4\}$,
where $i\neq j$, a segment in $g_i(P^{m})$ does not intersect any segments in $g_j(P^{m})$,
unless they are consecutive in $P^{m}$ and they intersect at a common endpoint.
This follows from the above claim together with the observation that for $i\neq j$, the intersection
$g_i(P^{m})\cap g_j(P^{m})$ is either empty or contains a single vertex which is the common
endpoint of two consecutive segments in $P^{m}$.
\end{proof}
In the remainder of this section, we assume that
\begin{equation}
a=\frac{c-2}{2c}=\frac{c_*}{2(c_*+1)}.
\end{equation}
Under this assumption, all segments in $P^1$ have the same length $a$.
Therefore, by construction, all segments in $P^m$ have the same length
\[
a^m=\left(\frac{c_*}{2(c_*+1)}\right)^m.
\]
There are $4^m$ segments in $P^m$, with $4^m+1$ vertices, and its stretch factor is
\[\delta_{P^m} = 4^m\left(\frac{c_*}{2(c_*+1)}\right)^m=\left(\frac{2c_*}{c_*+1}\right)^m.\]
Consequently, $m=\log_4 (n-1) = \frac{\log (n-1)}{2}$, and
\[
\delta_{P^m} =\left(\frac{2c_*}{c_*+1}\right)^{\frac{\log (n-1)}{2}}
=\left(\frac{2c-4}{c}\right)^{\frac{\log (n-1)}{2}}
=(n-1)^{\frac{1+\log (c-2) - \log c}{2}},
\]
as claimed.
To finish the proof of Theorem~\ref{thm:lower-bound}, it remains to show the
constructed polygonal chains are indeed
$c$-chains.

\begin{lemma}\label{lemma:c-chain}
For every constant $c\geq 4$, $\P_c$ is a family of $c$-chains.
\end{lemma}

We first prove a couple of facts that will be useful in the proof of Lemma~\ref{lemma:c-chain}.
We defer an intuitive explanation until after the formal statement of the
following lemma.
\begin{lemma}\label{lemma:technical}
Let $m \geq 1$ and let $P^m=(p_1, p_2, \dots, p_n)$, where $n=4^m+1$.
Then the following hold:
\begin{enumerate}[label=(\roman*), ref=\roman*]
\item\label{pf:f1}
There exists a sequence $(q_1, q_2, \dots, q_{\ell})$ of $\ell=2\cdot 4^{m-1}$ points in $\RR^2$ such
that the chain $R^m=(p_1, q_1, p_2, q_2, \dots, p_{\ell}, q_{\ell}, p_{\ell+1})$ is similar to $P^m$.
\item\label{pf:f2}
  For $m\geq 2$, define
  $g_5: \P_c\setminus\{P^0, P^1\} \to \P_c$ by
\[
g_5(P^m) = \left(g_3\circ g_2(P^m)\right) \cup \left(g_4\circ g_2(P^m)\right) \cup
\left(g_1\circ g_3(P^m)\right) \cup \left(g_2\circ g_3(P^m)\right).
\]
Then $g_5(P^m)$ is similar to $P^{m-1}$.
\end{enumerate}
\end{lemma}

Part~(\ref{pf:f1}) of Lemma~\ref{lemma:technical} says that given $P^m$,
we can construct a chain $R^m$ similar to $P^m$
by inserting one point between every two consecutive points of the left half of $P^m$,
see \textsc{Fig.}~\ref{fig:technical}\,(left).
Part~(\ref{pf:f2}) says that the ``top'' subchain of $P^m$ that consists of the right half of $g_2(P^m)$
and the left half of $g_3(P^m)$, see \textsc{Fig.}~\ref{fig:technical}\,(right), is similar to $P^{m-1}$.

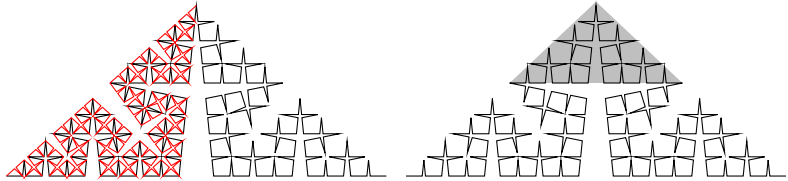
\begin{figure}[ht!]
\centering
\begin{tikzpicture}[scale=0.6]
\path (0,0) pic[scale=0.6] {p4};
\draw[red] (0.000, 0.000) -- (0.213, 0.193) -- (0.427, 0.000) -- (0.256, 0.232) -- (0.470, 0.425) -- (0.683, 0.232) -- (0.512, 0.000) -- (0.726, 0.193) -- (0.939, 0.000) -- (0.768, 0.232) -- (0.982, 0.425) -- (0.734, 0.278) -- (0.563, 0.510) -- (0.777, 0.703) -- (0.990, 0.510) -- (0.820, 0.741) -- (1.033, 0.934) -- (1.246, 0.741) -- (1.076, 0.510) -- (1.289, 0.703) -- (1.503, 0.510) -- (1.332, 0.278) -- (1.084, 0.425) -- (1.298, 0.232) -- (1.127, 0.000) -- (1.340, 0.193) -- (1.554, 0.000) -- (1.383, 0.232) -- (1.597, 0.425) -- (1.810, 0.232) -- (1.639, 0.000) -- (1.853, 0.193) -- (2.066, 0.000) -- (1.895, 0.232) -- (2.109, 0.425) -- (1.861, 0.278) -- (1.690, 0.510) -- (1.904, 0.703) -- (2.117, 0.510) -- (1.947, 0.741) -- (2.160, 0.934) -- (1.912, 0.788) -- (1.742, 1.019) -- (1.863, 0.758) -- (1.615, 0.612) -- (1.445, 0.843) -- (1.658, 1.036) -- (1.410, 0.890) -- (1.240, 1.121) -- (1.453, 1.314) -- (1.667, 1.121) -- (1.496, 1.353) -- (1.709, 1.546) -- (1.923, 1.353) -- (1.752, 1.121) -- (1.965, 1.314) -- (2.179, 1.121) -- (2.008, 1.353) -- (2.222, 1.546) -- (1.974, 1.399) -- (1.803, 1.631) -- (2.017, 1.824) -- (2.230, 1.631) -- (2.059, 1.863) -- (2.273, 2.056) -- (2.486, 1.863) -- (2.315, 1.631) -- (2.529, 1.824) -- (2.742, 1.631) -- (2.572, 1.399) -- (2.324, 1.546) -- (2.537, 1.353) -- (2.367, 1.121) -- (2.580, 1.314) -- (2.794, 1.121) -- (2.623, 1.353) -- (2.836, 1.546) -- (3.050, 1.353) -- (2.879, 1.121) -- (3.092, 1.314) -- (3.306, 1.121) -- (3.135, 0.890) -- (2.887, 1.036) -- (3.101, 0.843) -- (2.930, 0.612) -- (2.683, 0.758) -- (2.804, 1.019) -- (2.633, 0.788) -- (2.385, 0.934) -- (2.599, 0.741) -- (2.428, 0.510) -- (2.642, 0.703) -- (2.855, 0.510) -- (2.684, 0.278) -- (2.437, 0.425) -- (2.650, 0.232) -- (2.479, 0.000) -- (2.693, 0.193) -- (2.906, 0.000) -- (2.735, 0.232) -- (2.949, 0.425) -- (3.162, 0.232) -- (2.992, 0.000) -- (3.205, 0.193) -- (3.418, 0.000) -- (3.248, 0.232) -- (3.461, 0.425) -- (3.214, 0.278) -- (3.043, 0.510) -- (3.256, 0.703) -- (3.470, 0.510) -- (3.299, 0.741) -- (3.512, 0.934) -- (3.726, 0.741) -- (3.555, 0.510) -- (3.769, 0.703) -- (3.982, 0.510) -- (3.811, 0.278) -- (3.564, 0.425) -- (3.777, 0.232) -- (3.606, 0.000) -- (3.820, 0.193) -- (4.033, 0.000) -- (3.862, 0.232) -- (4.076, 0.425) -- (4.289, 0.232) -- (4.119, 0.000) -- (4.332, 0.193) -- (4.545, 0.000) -- (4.375, 0.232) -- (4.588, 0.425) -- (4.341, 0.278) -- (4.170, 0.510) -- (4.383, 0.703) -- (4.597, 0.510) -- (4.426, 0.741) -- (4.639, 0.934) -- (4.392, 0.788) -- (4.221, 1.019) -- (4.342, 0.758) -- (4.095, 0.612) -- (3.924, 0.843) -- (4.137, 1.036) -- (3.890, 0.890) -- (3.719, 1.121) -- (3.932, 1.314) -- (4.146, 1.121) -- (3.975, 1.353) -- (4.189, 1.546) -- (4.402, 1.353) -- (4.231, 1.121) -- (4.445, 1.314) -- (4.658, 1.121) -- (4.487, 1.353) -- (4.701, 1.546) -- (4.453, 1.399) -- (4.282, 1.631) -- (4.496, 1.824) -- (4.709, 1.631) -- (4.539, 1.863) -- (4.752, 2.056) -- (4.504, 1.909) -- (4.334, 2.141) -- (4.455, 1.880) -- (4.207, 1.733) -- (4.037, 1.965) -- (4.250, 2.158) -- (4.002, 2.011) -- (3.832, 2.243) -- (3.953, 1.982) -- (3.705, 1.835) -- (3.977, 1.929) -- (4.098, 1.668) -- (3.851, 1.522) -- (3.680, 1.753) -- (3.801, 1.492) -- (3.554, 1.346) -- (3.383, 1.577) -- (3.596, 1.770) -- (3.349, 1.624) -- (3.178, 1.855) -- (3.392, 2.048) -- (3.605, 1.855) -- (3.434, 2.087) -- (3.648, 2.280) -- (3.400, 2.133) -- (3.229, 2.365) -- (3.351, 2.104) -- (3.103, 1.957) -- (2.932, 2.189) -- (3.146, 2.382) -- (2.898, 2.235) -- (2.727, 2.467) -- (2.941, 2.660) -- (3.154, 2.467) -- (2.983, 2.699) -- (3.197, 2.892) -- (3.410, 2.699) -- (3.240, 2.467) -- (3.453, 2.660) -- (3.666, 2.467) -- (3.496, 2.699) -- (3.709, 2.892) -- (3.462, 2.745) -- (3.291, 2.977) -- (3.504, 3.170) -- (3.718, 2.977) -- (3.547, 3.208) -- (3.760, 3.401) -- (3.974, 3.208) -- (3.803, 2.977) -- (4.016, 3.170) -- (4.230, 2.977) -- (4.059, 2.745) -- (3.812, 2.892) -- (4.025, 2.699) -- (3.854, 2.467) -- (4.068, 2.660) -- (4.281, 2.467) -- (4.110, 2.699) -- (4.324, 2.892) -- (4.537, 2.699) -- (4.367, 2.467) -- (4.580, 2.660) -- (4.793, 2.467) -- (4.623, 2.699) -- (4.836, 2.892) -- (4.588, 2.745) -- (4.418, 2.977) -- (4.631, 3.170) -- (4.845, 2.977) -- (4.674, 3.208) -- (4.887, 3.401) -- (4.640, 3.255) -- (4.469, 3.486) -- (4.590, 3.225) -- (4.343, 3.079) -- (4.172, 3.310) -- (4.385, 3.503) -- (4.138, 3.357) -- (3.967, 3.588) -- (4.180, 3.781) -- (4.394, 3.588) -- (4.223, 3.820) -- (4.437, 4.013) -- (4.650, 3.820) -- (4.479, 3.588) -- (4.693, 3.781) -- (4.906, 3.588) -- (4.735, 3.820) -- (4.949, 4.013) -- (4.701, 3.866) -- (4.530, 4.098) -- (4.744, 4.291) -- (4.957, 4.098) -- (4.787, 4.330) -- (5.000, 4.523);

\end{tikzpicture}
\hfill
\begin{tikzpicture}[scale=0.6]
\fill[black!25] (2.727, 2.467) -- (5.000, 4.623) -- (7.273, 2.467) --cycle;
\path (0,0) pic[scale=0.6] {p4};
\end{tikzpicture}
\caption{Left: Chain $P^m$ with the scaled copy of itself $R^m$ (in red);
Right: Chain $P^m$ with its subchain $g_5(P^m)$ marked by its
convex hull.}\label{fig:technical}
\end{figure}

\begin{proof}[Proof of Lemma~\ref{lemma:technical}]
For part~(\ref{pf:f1}), we review the construction of $P^m$, and show that $R^m$ and $P^m$ can be
constructed in a coupled manner.
In \textsc{Fig.}~\ref{fig:induced-chain}\,(left), consider $P^1=(p_1, p_2, p_3, p_4, p_5)$.
Recall that all segments in $P^1$ are of the same length $a=\frac{c_*}{2(c_*+1)}$.
The isosceles triangles $\Delta p_1p_2p_3$ and $\Delta p_1p_3p_5$ are similar.
Let $\sigma: \Delta p_1p_3p_5 \to \Delta p_1p_2p_3$ be the similarity transformation.
Let $q_1=\sigma(p_2)$ and $q_2 = \sigma(p_4)$.
By construction, the chain $R^1 = (p_1, q_1, p_2, q_2, p_3)$ is similar to $P^1$.
In particular, all of its segments have the same length,
and so the isosceles triangle $\Delta p_1q_1p_2$ is similar to $\Delta p_1p_3p_5$.
Moreover, its base is the segment $p_1p_2$, so $\Delta p_1q_1p_2$ is precisely
$\conv(g_1(P^2))$, see \textsc{Fig.}~\ref{fig:induced-chain}\,(right).

\begin{figure}[!ht]
\centering
\begin{tikzpicture}[scale=0.6]
\draw (0.000, 0.000)node[circle, fill, inner sep=1pt, label=above:$p_1$]{}
-- (4.545, 0.000)node[circle, fill, inner sep=1pt, label=85:$p_2$]{}
-- (5.000, 4.523)node[circle, fill, inner sep=1pt, label=left:$p_3$]{}
-- (5.455, 0.000)node[circle, fill, inner sep=1pt, label=45:$p_4$]{}
-- (10.000, 0.000)node[circle, fill, inner sep=1pt, label=93:$p_5$]{};
\draw[red] (0.000, 0.000)
-- (2.273, 2.056)node[circle, fill, inner sep=1pt, label=left:$q_1$]{}
-- (4.545, 0.000)
-- (2.727, 2.467)node[circle, fill, inner sep=1pt, label=above:$q_2$]{}
-- (5.000, 4.523);
\end{tikzpicture}
\hfill
\begin{tikzpicture}[scale=0.6]
\draw[red] (0.000, 0.000)
-- (2.273, 2.056) -- (4.545, 0.000) -- (2.727, 2.467) -- (5.000, 4.523);
\draw (0.000, 0.000)node[circle, fill, inner sep=1pt, label=above:$v_1$]{}
-- (2.066, 0.000)node[circle, fill, inner sep=1pt, label=135:$v_2$]{}
-- (2.273, 2.056)node[circle, fill, inner sep=1pt, label=left:$v_{3}$]{}
-- (2.479, 0.000)node[circle, fill, inner sep=1pt, label=45:$v_{4}$]{}
-- (4.545, 0.000)node[circle, fill, inner sep=1pt, label=85:$v_{5}$]{}
-- (4.752, 2.056)node[circle, fill, inner sep=1pt, label=225:$v_{6}$]{}
-- (2.727, 2.467)node[circle, fill, inner sep=1pt, label=above:$v_{7}$]{}
-- (4.793, 2.467)node[circle, fill, inner sep=1pt, label=135:$v_{8}$]{}
-- (5.000, 4.523)node[circle, fill, inner sep=1pt, label=left:$v_{9}$]{}
-- (5.207, 2.467)node[circle, fill, inner sep=1pt, label=45:$v_{10}$]{}
-- (7.273, 2.467)node[circle, fill, inner sep=1pt, label=above:$v_{11}$]{}
-- (5.248, 2.056) -- (5.455, 0.000)
-- (7.521, 0.000) -- (7.727, 2.056) -- (7.934, 0.000)
-- (10.000, 0.000)node[circle, fill, inner sep=1pt, label=93:$v_{17}$]{};
\end{tikzpicture}
\caption{Left: the chains $P^1$ and $R^1$ (red);
Right: the chains $P^2$ and $R^1$ (red).}\label{fig:induced-chain}
\end{figure}
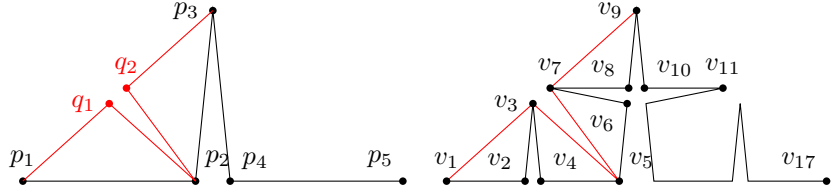

Write $P^2=(v_{1}, v_{2}, \dots, v_{17})$, then $v_{3}=q_1$ by the above argument
and $v_{7}=q_2$ by symmetry.
Now $\Delta v_{1}v_{2}v_{3}$, $\Delta v_{3}v_{4}v_{5}$, $\Delta v_{5}v_{6}v_{7}$, and
$\Delta v_{7}v_{8}v_{9}$ are four congruent isosceles triangles, all of which are
similar to $\Delta v_{1}v_{9}v_{17}$, since the angles are the same.
Repeat the above procedure on each of them to obtain
$R^2=(v_{1}, u_1, v_{2}, u_2, \dots, v_{8}, u_8, v_{9})$, which is similar to $P^2$.
Continue this construction inductively to get the desired chain $R^m$ for any
$m\geq 1$.

For part~(\ref{pf:f2}), see \textsc{Fig.}~\ref{fig:induced-chain}\,(right).
By definition, $g_5(P^2)$ is the subchain $(v_{7}, v_{8}, v_{9}, $ $v_{10}, v_{11})$.
Observe that the segments $v_{7}v_{8}$ and $v_{10}v_{11}$ are collinear by symmetry.
Moreover, they are parallel to $v_{1}v_{17}$ since
$\angle v_{7}v_{8}v_{9} = \angle v_{1}v_{5}v_{9}$.
So $g_5(P^2)$ is similar to $P^1$; see \textsc{Fig.}~\ref{fig:induced-chain}\,(left).
Then for $m\geq 2$, $g_5(P^m)$ is the subchain of $P^m$ starting at vertex $v_{7}$,
ending at vertex $v_{11}$.
By the construction of $P^m$, $g_5(P^m)$ is similar to $P^{m-1}$.
\end{proof}

\begin{proof}[Proof of Lemma~\ref{lemma:c-chain}]
We proceed by induction on $m$ again.
The claim is vacuously true for $P^0$.
For $P^1$, among all ten choices of $1\leq i<j<k\leq 5$,
$\frac{|p_2p_3|+|p_3p_4|}{|p_2p_4|}=c_*=\frac{c-2}{2}<c$ is the largest,
and so $P^1$ is also a $c$-chain.
Assume that $m\geq 2$ and $P^{m-1}$ is a $c$-chain.
We need to show that $P^m$ is also a $c$-chain.
Consider a triplet of vertices $\{p_i, p_j, p_k\}\subset P^m$, where $1\leq i< j< k\leq n=4^m+1$.

Recall that $P^m$ consists of four copies of the subchain $P^{m-1}$, namely $g_1(P^m)$, $g_2(P^m)$,
$g_3(P^m)$, and $g_4(P^m)$, see \textsc{Fig.}~\ref{fig:c-chain}\,(left).
If $\{p_i, p_j, p_k\}\subset g_l(P^{m})$ for any $l=1,2,3,4$, then by the induction hypothesis,
\[
\frac{|p_ip_j|+|p_jp_k|}{|p_ip_k|}\leq c.
\]
So we may assume that $p_i$ and $p_k$ belong to two different $g_l(P^m)$'s.
There are four cases to consider up to symmetry:
\begin{enumerate}[leftmargin=17mm,label=Case \arabic*., ref=Case~\arabic*]
\item\label{pf:c-chain-case1} $p_i\in g_1(P^m)$ and $p_k\in g_2(P^m)$;
\item $p_i\in g_1(P^m)$ and $p_k\in g_3(P^m)$;
\item $p_i\in g_1(P^m)$ and $p_k\in g_4(P^m)$;
\item\label{pf:c-chain-case4} $p_i\in g_2(P^m)$ and $p_k\in g_3(P^m)$.
\end{enumerate}

\begin{figure}[htbp!]
 \centering
\begin{tikzpicture}[scale=0.6]
\fill[black!25] (0.000, 0.000) -- (2.273, 2.056) -- (4.545, 0.000) --cycle;
\fill[black!25] (4.545, 0.000) -- (2.727, 2.467) -- (5.000, 4.623) --cycle;
\fill[black!25] (5.000, 4.623) -- (7.273, 2.467) -- (5.455, 0.000) --cycle;
\fill[black!25] (5.455, 0.000) -- (7.727, 2.056) -- (10.000, 0.000) --cycle;
\draw[dashed] (4.545, 0.000) -- node[below]{$\frac{1}{c_*+1}$} (5.455, 0.000);
\path (0,0) pic[scale=0.6] {p4};
\end{tikzpicture}
\hfill
\begin{tikzpicture}[scale=0.6]
\path (0,0) pic[scale=0.6] {p4};
\draw[red] (0.000, 0.000) -- (0.213, 0.193) -- (0.427, 0.000) -- (0.256, 0.232) -- (0.470, 0.425) -- (0.683, 0.232) -- (0.512, 0.000) -- (0.726, 0.193) -- (0.939, 0.000) -- (0.768, 0.232) -- (0.982, 0.425) -- (0.734, 0.278) -- (0.563, 0.510) -- (0.777, 0.703) -- (0.990, 0.510) -- (0.820, 0.741) -- (1.033, 0.934) -- (1.246, 0.741) -- (1.076, 0.510) -- (1.289, 0.703) -- (1.503, 0.510) -- (1.332, 0.278) -- (1.084, 0.425) -- (1.298, 0.232) -- (1.127, 0.000) -- (1.340, 0.193) -- (1.554, 0.000) -- (1.383, 0.232) -- (1.597, 0.425) -- (1.810, 0.232) -- (1.639, 0.000) -- (1.853, 0.193) -- (2.066, 0.000) -- (1.895, 0.232) -- (2.109, 0.425) -- (1.861, 0.278) -- (1.690, 0.510) -- (1.904, 0.703) -- (2.117, 0.510) -- (1.947, 0.741) -- (2.160, 0.934) -- (1.912, 0.788) -- (1.742, 1.019) -- (1.863, 0.758) -- (1.615, 0.612) -- (1.445, 0.843) -- (1.658, 1.036) -- (1.410, 0.890) -- (1.240, 1.121) -- (1.453, 1.314) -- (1.667, 1.121) -- (1.496, 1.353) -- (1.709, 1.546) -- (1.923, 1.353) -- (1.752, 1.121) -- (1.965, 1.314) -- (2.179, 1.121) -- (2.008, 1.353) -- (2.222, 1.546) -- (1.974, 1.399) -- (1.803, 1.631) -- (2.017, 1.824) -- (2.230, 1.631) -- (2.059, 1.863) -- (2.273, 2.056) -- (2.486, 1.863) -- (2.315, 1.631) -- (2.529, 1.824) -- (2.742, 1.631) -- (2.572, 1.399) -- (2.324, 1.546) -- (2.537, 1.353) -- (2.367, 1.121) -- (2.580, 1.314) -- (2.794, 1.121) -- (2.623, 1.353) -- (2.836, 1.546) -- (3.050, 1.353) -- (2.879, 1.121) -- (3.092, 1.314) -- (3.306, 1.121) -- (3.135, 0.890) -- (2.887, 1.036) -- (3.101, 0.843) -- (2.930, 0.612) -- (2.683, 0.758) -- (2.804, 1.019) -- (2.633, 0.788) -- (2.385, 0.934) -- (2.599, 0.741) -- (2.428, 0.510) -- (2.642, 0.703) -- (2.855, 0.510) -- (2.684, 0.278) -- (2.437, 0.425) -- (2.650, 0.232) -- (2.479, 0.000) -- (2.693, 0.193) -- (2.906, 0.000) -- (2.735, 0.232) -- (2.949, 0.425) -- (3.162, 0.232) -- (2.992, 0.000) -- (3.205, 0.193) -- (3.418, 0.000) -- (3.248, 0.232) -- (3.461, 0.425) -- (3.214, 0.278) -- (3.043, 0.510) -- (3.256, 0.703) -- (3.470, 0.510) -- (3.299, 0.741) -- (3.512, 0.934) -- (3.726, 0.741) -- (3.555, 0.510) -- (3.769, 0.703) -- (3.982, 0.510) -- (3.811, 0.278) -- (3.564, 0.425) -- (3.777, 0.232) -- (3.606, 0.000) -- (3.820, 0.193) -- (4.033, 0.000) -- (3.862, 0.232) -- (4.076, 0.425) -- (4.289, 0.232) -- (4.119, 0.000) -- (4.332, 0.193) -- (4.545, 0.000) -- (4.375, 0.232) -- (4.588, 0.425) -- (4.341, 0.278) -- (4.170, 0.510) -- (4.383, 0.703) -- (4.597, 0.510) -- (4.426, 0.741) -- (4.639, 0.934) -- (4.392, 0.788) -- (4.221, 1.019) -- (4.342, 0.758) -- (4.095, 0.612) -- (3.924, 0.843) -- (4.137, 1.036) -- (3.890, 0.890) -- (3.719, 1.121) -- (3.932, 1.314) -- (4.146, 1.121) -- (3.975, 1.353) -- (4.189, 1.546) -- (4.402, 1.353) -- (4.231, 1.121) -- (4.445, 1.314) -- (4.658, 1.121) -- (4.487, 1.353) -- (4.701, 1.546) -- (4.453, 1.399) -- (4.282, 1.631) -- (4.496, 1.824) -- (4.709, 1.631) -- (4.539, 1.863) -- (4.752, 2.056) -- (4.504, 1.909) -- (4.334, 2.141) -- (4.455, 1.880) -- (4.207, 1.733) -- (4.037, 1.965) -- (4.250, 2.158) -- (4.002, 2.011) -- (3.832, 2.243) -- (3.953, 1.982) -- (3.705, 1.835) -- (3.977, 1.929) -- (4.098, 1.668) -- (3.851, 1.522) -- (3.680, 1.753) -- (3.801, 1.492) -- (3.554, 1.346) -- (3.383, 1.577) -- (3.596, 1.770) -- (3.349, 1.624) -- (3.178, 1.855) -- (3.392, 2.048) -- (3.605, 1.855) -- (3.434, 2.087) -- (3.648, 2.280) -- (3.400, 2.133) -- (3.229, 2.365) -- (3.351, 2.104) -- (3.103, 1.957) -- (2.932, 2.189) -- (3.146, 2.382) -- (2.898, 2.235) -- (2.727, 2.467) -- (2.941, 2.660) -- (3.154, 2.467) -- (2.983, 2.699) -- (3.197, 2.892) -- (3.410, 2.699) -- (3.240, 2.467) -- (3.453, 2.660) -- (3.666, 2.467) -- (3.496, 2.699) -- (3.709, 2.892) -- (3.462, 2.745) -- (3.291, 2.977) -- (3.504, 3.170) -- (3.718, 2.977) -- (3.547, 3.208) -- (3.760, 3.401) -- (3.974, 3.208) -- (3.803, 2.977) -- (4.016, 3.170) -- (4.230, 2.977) -- (4.059, 2.745) -- (3.812, 2.892) -- (4.025, 2.699) -- (3.854, 2.467) -- (4.068, 2.660) -- (4.281, 2.467) -- (4.110, 2.699) -- (4.324, 2.892) -- (4.537, 2.699) -- (4.367, 2.467) -- (4.580, 2.660) -- (4.793, 2.467) -- (4.623, 2.699) -- (4.836, 2.892) -- (4.588, 2.745) -- (4.418, 2.977) -- (4.631, 3.170) -- (4.845, 2.977) -- (4.674, 3.208) -- (4.887, 3.401) -- (4.640, 3.255) -- (4.469, 3.486) -- (4.590, 3.225) -- (4.343, 3.079) -- (4.172, 3.310) -- (4.385, 3.503) -- (4.138, 3.357) -- (3.967, 3.588) -- (4.180, 3.781) -- (4.394, 3.588) -- (4.223, 3.820) -- (4.437, 4.013) -- (4.650, 3.820) -- (4.479, 3.588) -- (4.693, 3.781) -- (4.906, 3.588) -- (4.735, 3.820) -- (4.949, 4.013) -- (4.701, 3.866) -- (4.530, 4.098) -- (4.744, 4.291) -- (4.957, 4.098) -- (4.787, 4.330) -- (5.000, 4.523);
\node[below, inner sep=3.43mm]{} (5.455, 0.000);
\end{tikzpicture}
\caption{Left: Chain $P^m$ with its four subchains of type $P^{m-1}$ marked
by their convex hulls;
Right: Chain $P^m$ with the scaled copy of itself $R^m$ (in red) constructed in
   Lemma~\ref{lemma:technical}\,(\ref{pf:f1}).}\label{fig:c-chain}
\end{figure}
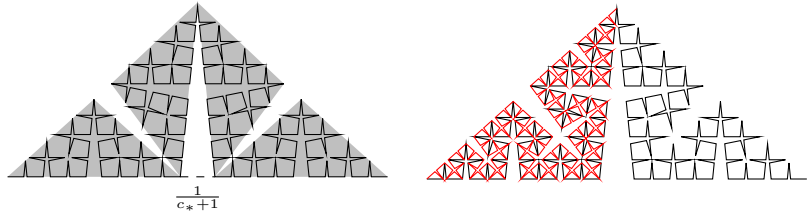

By Lemma~\ref{lemma:technical}\,(\ref{pf:f1}), the vertex set of $g_1(P^m)\cup g_2(P^m)$
is contained in the chain $R^m$ shown in \textsc{Fig.}~\ref{fig:c-chain}\,(right).
If we are in~\ref{pf:c-chain-case1}, \ie, $p_i\in g_1(P^{m})$ and $p_k\in g_2(P^{m})$,
then $p_i, p_j, p_k$ can be thought of as vertices of $R^m$.
The similarity between $R^m$ and $P^m$, maps points $p_i, p_j, p_k$ to
suitable points $p_i', p_j', p_k' \in P^m$ such that
\[
\frac{|p_i'p_j'|+|p_j'p_k'|}{|p_i'p_k'|} =\frac{|p_ip_j|+|p_jp_k|}{|p_ip_k|}.
\]
Since $p_i\in g_1(R^{m}) \cup g_2(R^{m})$ while $p_k\in g_3(R^{m}) \cup g_4(R^{m})$, the triplet
$(p_i', p_j', p_k')$ does not belong to \ref{pf:c-chain-case1}.
In other words, \ref{pf:c-chain-case1} can be represented by other cases.

Recall that in Lemma~\ref{lemma:simple}, we showed that $\conv(P^m)$ is
an isosceles triangle $T$ of diameter $1$.
Observe that if $|p_ip_k|\geq \frac{1}{c_*+1}$,
then
\[
\frac{|p_ip_j|+|p_jp_k|}{|p_ip_k|}\leq \frac{1+1}{\frac{1}{c_*+1}}= 2c_*+2=c,
\]
as required.
So we may assume that $|p_ip_k|<\frac{1}{c_*+1}$, therefore only \ref{pf:c-chain-case4} remains,
\ie, $p_i\in g_2(P^{m})$ and $p_k\in g_3(P^{m})$.

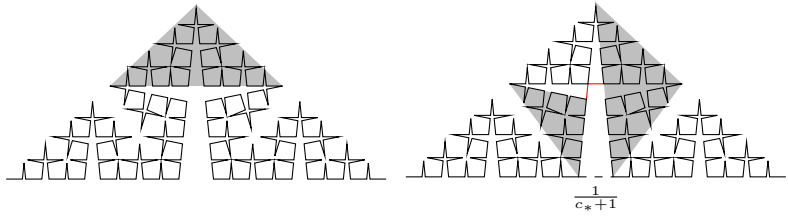
\begin{figure}[htbp!]
\centering
\begin{tikzpicture}[scale=0.6]
\fill[black!25] (2.727, 2.467) -- (5.000, 4.623) -- (7.273, 2.467) --cycle;
\path (0,0) pic[scale=0.6] {p4};
\node[below, inner sep=3.43mm]{} (5.455, 0.000);
\end{tikzpicture}
\hfill
\begin{tikzpicture}[scale=0.6]
\fill[black!25] (4.545, 0.000) -- (2.727, 2.467) -- (4.752, 2.056)--cycle;
\fill[black!25] (5.000, 4.623) -- (7.273, 2.467) -- (5.455, 0.000) --cycle;
\draw[red] (4.752, 2.056) -- (4.793, 2.467) -- (5.207, 2.467);
\draw[dashed] (4.545, 0.000) -- node[below]{$\frac{1}{c_*+1}$} (5.455, 0.000);

\path (0,0) pic[scale=0.6] {p4};
\end{tikzpicture}
\caption{Left: Chain $P^m$ with its subchain $g_5(P^{m})$ marked by its convex hull;
Right: The last case where $p_i$ is in the left shaded subchain and $p_k$ is in the right shaded subchain.
}\label{fig:c-chain2}
\end{figure}

By Lemma~\ref{lemma:technical}\,(\ref{pf:f2}), the ``top'' subchain $g_5(P^m)$ of $P^m$
is also similar to $P^{m-1}$, see \textsc{Fig.}~\ref{fig:c-chain2}\,(left).
If $p_i$ and $p_k$ are both in $g_5(P^{m})$,
\ie, $p_i\in \left(g_3\circ g_2(P^m)\right) \cup \left(g_4\circ g_2(P^m)\right)$
and $p_k\in \left(g_1\circ g_3(P^m)\right) \cup \left(g_2\circ g_3(P^m)\right)$, then so is $p_j$.

By the induction hypothesis, we have
\[
\frac{|p_ip_j|+|p_jp_k|}{|p_ip_k|}\leq c.
\]
So we may assume that at least one of $p_i$ and $p_k$ is not in $g_5(P^{m})$.
Without loss of generality, let $p_i\in g_2(P^{m}) \setminus g_5(P^{m})$.
The similarities that map $P^{m-1}$ to $g_2(P^{m})$ and $g_5(P^{m})$, respectively, have the same
scaling factor of $a=\frac{c_*}{2(c_*+1)}$, and they carry the bottom dashed segment in
\textsc{Fig.}~\ref{fig:c-chain2}\,(right), to the two red segments.

\begin{claim*} If $p_i\in g_2(P^{m}) \setminus g_5(P^{m})$
and $p_k\in g_3(P^{m})$, then $|p_ip_k|>\frac{c_*}{2(c_*+1)^2}$.
\end{claim*}

\begin{proof}
As noted above, we assume that $p_i$ is in $\conv(g_2(P^{m}) \setminus g_5(P^{m}))=
\Delta q_1q_2q_3$ in \textsc{Fig.}~\ref{fig:c-chain-case4}.
If $p_k\in g_5(P^m) \cap g_3(P^m)=\Delta q_7q_6q_5$, then the configuration is illustrated in
\textsc{Fig.}~\ref{fig:c-chain-case4}\,(left).
Note that $\Delta q_1q_2q_3$ and $\Delta q_7q_6q_5$ are reflections of each other with respect to
the bisector of $\angle q_3q_4q_5$.
Hence the shortest distance between $\Delta q_1q_2q_3$ and $\Delta q_7q_6q_5$ is
$\min\{|q_3q_5|, |q_2q_6|, |q_1q_7|\}$. Since $c_*\geq 1$, we have
\[|q_1q_7|>|q_7q_9|=|q_3q_5|=a^{3/2}=\left(\frac{c_*}{2(c_*+1)}\right)^{3/2}\geq \frac{c_*}{2(c_*+1)^2}.\]
Further note that $q_2q_4q_6q_8$ is an isosceles trapezoid, so the length of its diagonal is bounded by
$|q_2q_6|>|q_2q_4|=\frac{c_*}{2(c_*+1)^2}$.
Therefore the claim holds when $p_k\in\Delta q_7q_6q_5$.

Otherwise $p_k\in g_3(P^m) \setminus g_5(P^m)=\Delta q_9q_8q_7$: see
\textsc{Fig.}~\ref{fig:c-chain-case4}\,(right).
Note that $\Delta q_1q_2q_3$ and $\Delta q_9q_8q_7$ are reflections of each other with respect to
the bisector of $\angle q_4q_5q_6$.
So the shortest distance between the shaded triangles is
the minimum between $|q_3q_7|$, $|q_2q_8|$, and $|q_1q_9|$.
However, all three candidates are strictly larger than $|q_4q_6|=\frac{c_*}{2(c_*+1)^2}$.
This completes the proof of the claim.
\end{proof}
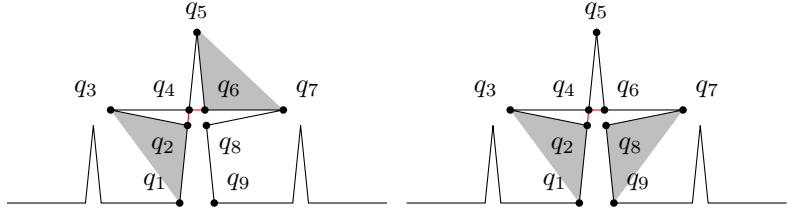
\begin{figure}[!ht]
\centering
\begin{tikzpicture}[scale=0.6]
\fill[black!25] (4.545, 0.000) -- (2.727, 2.467) -- (4.752, 2.056)--cycle;
\fill[black!25] (5.000, 4.623) -- (5.207, 2.467) -- (7.273, 2.467)--cycle;
\draw[red] (4.752, 2.056) -- (4.793, 2.467) -- (5.207, 2.467);
\draw (0.000, 0.000) -- (2.066, 0.000)
-- (2.273, 2.056) -- (2.479, 0.000)
-- (4.545, 0.000)node[circle, fill, inner sep=1pt, label=135:$q_1$]{}
-- (4.752, 2.056)node[circle, fill, inner sep=1pt, label=225:$q_2$]{}
-- (2.727, 2.467)node[circle, fill, inner sep=1pt, label=135:$q_3$]{}
-- (4.793, 2.467)node[circle, fill, inner sep=1pt, label=135:$q_4$]{}
-- (5.000, 4.523)node[circle, fill, inner sep=1pt, label=above:$q_5$]{}
-- (5.207, 2.467)node[circle, fill, inner sep=1pt, label=45:$q_6$]{}
-- (7.273, 2.467)node[circle, fill, inner sep=1pt, label=45:$q_7$]{}
-- (5.248, 2.056)node[circle, fill, inner sep=1pt, label=-45:$q_8$]{}
-- (5.455, 0.000)node[circle, fill, inner sep=1pt, label=45:$q_9$]{}
-- (7.521, 0.000) -- (7.727, 2.056) -- (7.934, 0.000) -- (10.000, 0.000);
\end{tikzpicture}
\hfill
\begin{tikzpicture}[scale=0.6]
\fill[black!25] (4.545, 0.000) -- (2.727, 2.467) -- (4.752, 2.056)--cycle;
\fill[black!25] (7.273, 2.467) -- (5.248, 2.056) -- (5.455, 0.000)--cycle;
\draw[red] (4.752, 2.056) -- (4.793, 2.467) -- (5.207, 2.467);
\draw (0.000, 0.000) -- (2.066, 0.000)
-- (2.273, 2.056) -- (2.479, 0.000)
-- (4.545, 0.000)node[circle, fill, inner sep=1pt, label=135:$q_1$]{}
-- (4.752, 2.056)node[circle, fill, inner sep=1pt, label=225:$q_2$]{}
-- (2.727, 2.467)node[circle, fill, inner sep=1pt, label=135:$q_3$]{}
-- (4.793, 2.467)node[circle, fill, inner sep=1pt, label=135:$q_4$]{}
-- (5.000, 4.523)node[circle, fill, inner sep=1pt, label=above:$q_5$]{}
-- (5.207, 2.467)node[circle, fill, inner sep=1pt, label=45:$q_6$]{}
-- (7.273, 2.467)node[circle, fill, inner sep=1pt, label=45:$q_7$]{}
-- (5.248, 2.056)node[circle, fill, inner sep=1pt, label=-45:$q_8$]{}
-- (5.455, 0.000)node[circle, fill, inner sep=1pt, label=45:$q_9$]{}
-- (7.521, 0.000) -- (7.727, 2.056) -- (7.934, 0.000) -- (10.000, 0.000);
\end{tikzpicture}
\caption{$p_i\in\Delta q_1q_2q_3$,
  Left: $p_k\in\Delta q_7q_6q_5$;
  Right: $p_k\in \Delta q_9q_8q_7$.}\label{fig:c-chain-case4}
\end{figure}

Now the diameter of $g_2(P^{m}) \cup g_3(P^{m})$ is $a=\frac{c_*}{2(c_*+1)}$
(note that there are three diameter pairs), so
\[
\frac{|p_ip_j|+|p_jp_k|}{|p_ip_k|}<\frac{2\cdot \frac{c_*}{2(c_*+1)}}{\frac{c_*}{2(c_*+1)^2}}= 2c_*+2=c,
\]
as required.
This concludes the proof of Lemma~\ref{lemma:c-chain} and Theorem~\ref{thm:lower-bound}.
\end{proof}

\section{Generalizations to Higher Dimensions}\label{sec:higher-dim}

A $c$-chain $P$ with $n$ vertices and its stretch factor $\delta_P$
can be defined in any metric space, not just the Euclidean
plane.
We now discuss how our results generalize to other metric spaces,
with a particular focus on the high-dimensional Euclidean space
$\RR^d$.
First,  we examine the
upper bounds from Section~\ref{sec:upper}.

\subsection{Upper bounds}
As already noted in Section~\ref{sec:upper}, the upper bound
$\delta_P \leq c(n-1)^{\log c}$ of Theorem~\ref{thm:logc}
holds for any positive distance function that need not even
satisfy the triangle inequality.

Theorem~\ref{thm:linear} uses only the triangle inequality,
and the bound $\delta_P\leq c(n-2)+1$ holds in any metric space.
This bound cannot be improved,
in the following sense:
For every $c \geq 2+\sqrt{5}$ and even $n$, we can define a finite
metric space on the vertex
set of $P$ by $|p_1p_n|=1$; for $1<i<n$,
\[
|p_1p_i|=
\begin{cases}
  \frac{c+1}{2} & \text{if } i \text{ is even}\\
  \frac{c-1}{2} & \text{if } i \text{ is odd}
\end{cases}
\text{ and }
|p_ip_n|=
\begin{cases}
  \frac{c-1}{2} & \text{if } i \text{ is even}\\
  \frac{c+1}{2} & \text{if } i \text{ is odd}
\end{cases};
\]
and $|p_ip_j|=c$ for all $1<i<j<n$.
It is easy to verify that $P$ is a $c$-chain (the case that
puts the strongest constraint on $c$ in (\ref{eq:c-chain})
occurs if, e.g., $i = 1$, $1 < j < n$ is even, and $j < k < n$
is odd) and that $P$ has stretch factor
\[
  \delta_P = \frac{\sum_{i=1}^{n-1}|p_ip_{i+1}|}{|p_1p_n|}=|p_1p_2|+|p_{n-1}p_n|+\sum_{i=2}^{n-2}|p_ip_{i+1}|=c(n-2)+1.
\]

The proof of Theorem~\ref{thm:1/2} uses a volume argument in the plane.
The argument extends to $\mathbb{R}^d$, for all constant dimensions
$d\geq 2$, and yields $\delta_P =O\left(c^2(n-1)^{(d-1)/d}\right)$.
\begin{theorem}\label{thm:dspace}
For a $c$-chain $P$ with $n$ vertices in $\mathbb{R}^d$,
for some constant $d\geq 2$, we have
$$\delta_P =O\left(c^2(n-1)^{(d-1)/d}\right).$$
\end{theorem}
\begin{proof}
Let $P=(p_1,\dots, p_n)$ be a $c$-chain in $\mathbb{R}^d$,
for some constants $c \geq 1$ and $d\in \mathbb{N}$.
We may assume that $|p_1 p_n|=1$.
By the $c$-chain property, all vertices of $P$ lie in an ellipsoid $E$ with
foci at $p_1$ and $p_n$, with major axis of length $c$.
Let $U$ be a ball of radius $c/2$ concentric with $E$; and note that $E\subseteq U$.

We set $x=c^2/(n-1)^{1/d}$; and let $L_0$ and $L_1$ be the sum of lengths of all edges in $P$
of length at most $x$ and more than $x$, respectively. By definition, we have $L=L_0+L_1$ and
\begin{equation}\label{eq:shortDD}
L_0\leq (n-1)x =c^2(n-1)^{(d-1)/d}.
\end{equation}
We shall prove that $L_1= O\left(c^2(n-1)^{(d-1)/d}\right)$.
For this, we further classify the edges in $L_1$ according to their lengths:
For $\ell=0,1,\dots , \infty$, let
\begin{equation}\label{eq:PL1DD}
P_\ell=\left\{p_i: 2^\ell x< |p_ip_{i+1}|\leq 2^{\ell+1}x\right\}.
\end{equation}
As shown in the proof of Theorem~\ref{thm:linear},
we have $|p_ip_{i+1}|\leq c$, for all $i=0,\dots, n-1$.
Consequently, $P_\ell=\emptyset$ when $c\leq 2^\ell x$,
or equivalently $\log(c/x)\leq \ell$.

We use a volume argument to derive an upper bound on the cardinality of $P_\ell$, for
$\ell=0,1,\dots , \lfloor\log (c/x)\rfloor$.
Assume that $p_i,p_k\in P_\ell$, and w.l.o.g., $i<k$.
If $k=i+1$, then $2^\ell x<|p_ip_k|$ by~\eqref{eq:PL1DD}.
Otherwise,
\[
2^\ell x<|p_ip_{i+1}| < |p_ip_{i+1}|+|p_{i+1}p_k| \leq c|p_ip_k|,
\text{  or  } \frac{2^\ell x}{c} < |p_ip_k|.
\]
Consequently, the balls of radius
\begin{equation}\label{area:RaDD}
R=\frac{2^\ell x}{2c} = \frac{2^\ell c}{2(n-1)^{1/d}}
\end{equation}
centered at the points in $P_{\ell}$ are interior-disjoint.
The volume of each ball is $\alpha_d R^d$, where $\alpha_d>0$ depends on $d$ only.
Since $P_\ell\subset E$, these balls are contained in the $R$-neighborhood of
the ball $U$, which is a ball $U_R$ of radius $\frac{c}{2}+R$ concentric with $U$.
For $\ell \leq \log(c/x)$, we have $2^\ell x\leq c$, hence
$R=\frac{2^\ell x}{2c} \leq \frac{c}{2c}=\frac12$.
Consequently, the radius of $U_R$ is at most $c$.
Since $U_R$ contains $|P_\ell|$ interior-disjoint balls
of radius $R$, we obtain
\begin{equation}\label{eq:PL2DD}
|P_\ell|
\leq \frac{\alpha_d c^d}{\alpha_d R^d}
= \left(\frac{c}{R}\right)^d
=\left(\frac{2(n-1)^{1/d}}{2^\ell}\right)^d
\leq \frac{2^d}{2^{d\ell}} (n-1).
\end{equation}
For every segment $p_{i}p_{i+1}$ with length more than $x$, we have that $p_i\in P_\ell$,
for some $\ell\in \{0,1,\dots , \lfloor\log (c/x)\rfloor\}$.
Using (\ref{eq:PL2DD}), the total length of these segments is
\begin{align*}
L_1&\leq  \sum_{\ell=0}^{\lfloor\log (c/x)\rfloor} |P_\ell| \cdot 2^{\ell+1}x
 <   \sum_{\ell=0}^{\lfloor\log (c/x)\rfloor}  \frac{2^d}{2^{d\ell}} (n-1) \cdot 2^{\ell+1}\cdot \frac{c^2}{(n-1)^{1/d}}\\
&<   2^{d+1}c^2(n-1)^{\frac{d-1}{d}} \sum_{\ell=0}^\infty \frac{1}{2^{(d-1)\ell}}
\leq  2^{d+2}c^2(n-1)^{(d-1)/d},
\end{align*}
as required. Together with \eqref{eq:shortDD}, this yields
$L= O\left(c^2(n-1)^{(d-1)/d}\right)$.
\end{proof}

\subsection{Lower bounds in \texorpdfstring{$\mathbb{R}^d$}{higher dimensions}}
We show that the exponent $(d-1)/d$ in Theorem~\ref{thm:dspace}
cannot be improved. More precisely, for
every $\eps>0$, we construct a family of axis-parallel chains in $\RR^d$
whose stretch factor is $n^{(1-\eps)(d-1)/d}$ for sufficiently large $n(\eps)$.
For the higher-dimensional case, we focus on axis-parallel chains, as
they are easier to analyze.
In the plane ($d=2$), this construction is also possible, but it yields weaker bounds than Theorem~\ref{thm:lower-bound}.

\begin{theorem}\label{thm:d-lower-bound}
Let $d\geq 2$ be an integer. For all constants
$\eps>0$ and sufficiently large $c = \Omega(d)$,
there is a positive integer $n_0$ such that for every $n\geq n_0$,
there exists an axis-parallel $c$-chain in $\RR^d$ with
$n$ vertices and stretch factor at least $(n-1)^{(1-\eps)(d-1)/d}$.
\end{theorem}
\begin{proof}
Let $d\geq 2$, $\eps>0$, and $c= \Omega(d)$ be given.
We describe a recursive construction in terms of
an even integer parameter 
\begin{equation}\label{eq:r}
r > 3^{(1-\eps)/(d\eps)}.
\end{equation}
We recursively define a family $\mathcal{Q}_c=\{Q^m\}_{m\in \NN}$
of axis-parallel $c$-chains in $\RR^d$, where each chain $Q^m$ has
$n_m\leq 3^{m+1} r^{dm}$ vertices.
Then, we show that the stretch factor of every $Q^m$ is at least
$(n_m-1)^{(1-\eps)(d-1)/d}$ for sufficiently large $m\in \NN$.

\subparagraph*{Construction of $\mathcal{Q}_c$.}
For each chain in $\mathcal{Q}_c$, we maintain a subset of \emph{active} directed
edges, which are disjoint, have the same length, and are parallel to the same coordinate axis.
In a nutshell, the recursion works as follows.
We start with a chain $Q^0$ that consists
of a single segment that is labeled active;
then for $m=1, 2, \ldots$, we obtain $Q^{m}$ by replacing each active
edge in a fixed chain $\pi$ by a homothetic copy of $Q^{m-1}$.
The chain $\pi$ is defined below;
it consists of $6r^d+1$ edges, $3r^d$ of which are active.

We define the chain $\pi$ in four steps, see Fig.~\ref{fig:pi} for an illustration.
Let $\mathbf{e_i}$, $i=1,\ldots , d$, be the standard basis vectors in $\mathbb{R}^d$.
\begin{enumerate}[label=(\arabic*)]
\item Consider the $(d-1)$-dimensional hyperrectangle $A=[0,1]\times [0,r-1]^{d-2}$.
Let $\gamma_0$ be an axis-parallel Hamiltonian cycle on the $2r^{d-2}$ integer
points that lie in $A$ such that the origin is incident to an edge
parallel to the $x_1$-axis.
We label the vertices of $\gamma_0$ by $v_i$, for $i=1, \ldots, 2r^{d-2}$, in
order,
where $v_1$ is the origin.
\item Let $a=(3r^2+1)/(3r)=r+1/(3r)$, and consider the $d$-dimensional
hyperrectangle $A\times [0,a]=[0,1]\times [0,r-1]^{d-2}\times [0,a]$.
We construct a Hamiltonian cycle $\gamma_1$
on the $4r^{d-2}$ points in
\[\left\{v_i\times\{0, a\} \mid i=1, \ldots, 2r^{d-2} \right\}\]
by replacing every edge $(v_{2i-1}, v_{2i})$ in $\gamma_0$ with three edges
\[
((v_{2i-1},0), (v_{2i-1}, a)),\;
((v_{2i-1},a), (v_{2i}, a)),\;
\text{and } ((v_{2i},a), (v_{2i},0)).
\]
Note that $\gamma_1$ has $4r^{d-2}$ edges, such that
$2r^{d-2}$ edges have length $a$ and are parallel to the $x_d$-axis.
Also note that the origin $v_1$ is incident to a unit edge parallel to the $x_1$-axis,
and to an edge of length $a$ parallel to the $x_d$-axis.
\item Delete the edge of $\gamma_1$ that is
incident to the origin $v_1$ and parallel to the $x_1$-axis.
This turns $\gamma_1$
into a Hamiltonian chain $\gamma_2$ from the origin to the vertex $\mathbf{e}_1$
in the hyperrectangle $A\times [0,a]=[0,1]\times [0,r-1]^{d-2}\times [0,a]$.
\item Consider the hyperrectangle $B(\pi)=\left[0,3r^2+1\right]\times [0,r-1]^{d-2}\times [0,a]$.
Let $\pi$ be the chain from the origin to $(3r^2+1)\cdot \mathbf{e}_1$
that is obtained by the concatenation of $3r^2/2$ copies of $\gamma_2$,
translated by vectors $(2j-1)\cdot \mathbf{e}_1$ for $j=1,2,\ldots , 3r^2/2$,
interlaced with $3r^2/2+1$ unit segments parallel to $\mathbf{e}_1$.
Note that $\pi$ has $\left(3r^2/2\right)\cdot \left(4r^{d-2}-1\right)+3r^2/2+1= 6r^d+1$ edges,
$\left(3r^2/2\right)\cdot 2r^{d-2}=3r^d$ of which have length $a$
and are parallel to the $x_d$-axis.
We label all these edges as active, so that $\pi$ has $3r^d$ active edges.
Observe that $B(\pi)$ is the minimum axis-parallel bounding box of $\pi$.
\end{enumerate}

\begin{figure}[htbp]
\centering
\includegraphics[width=0.75\textwidth]{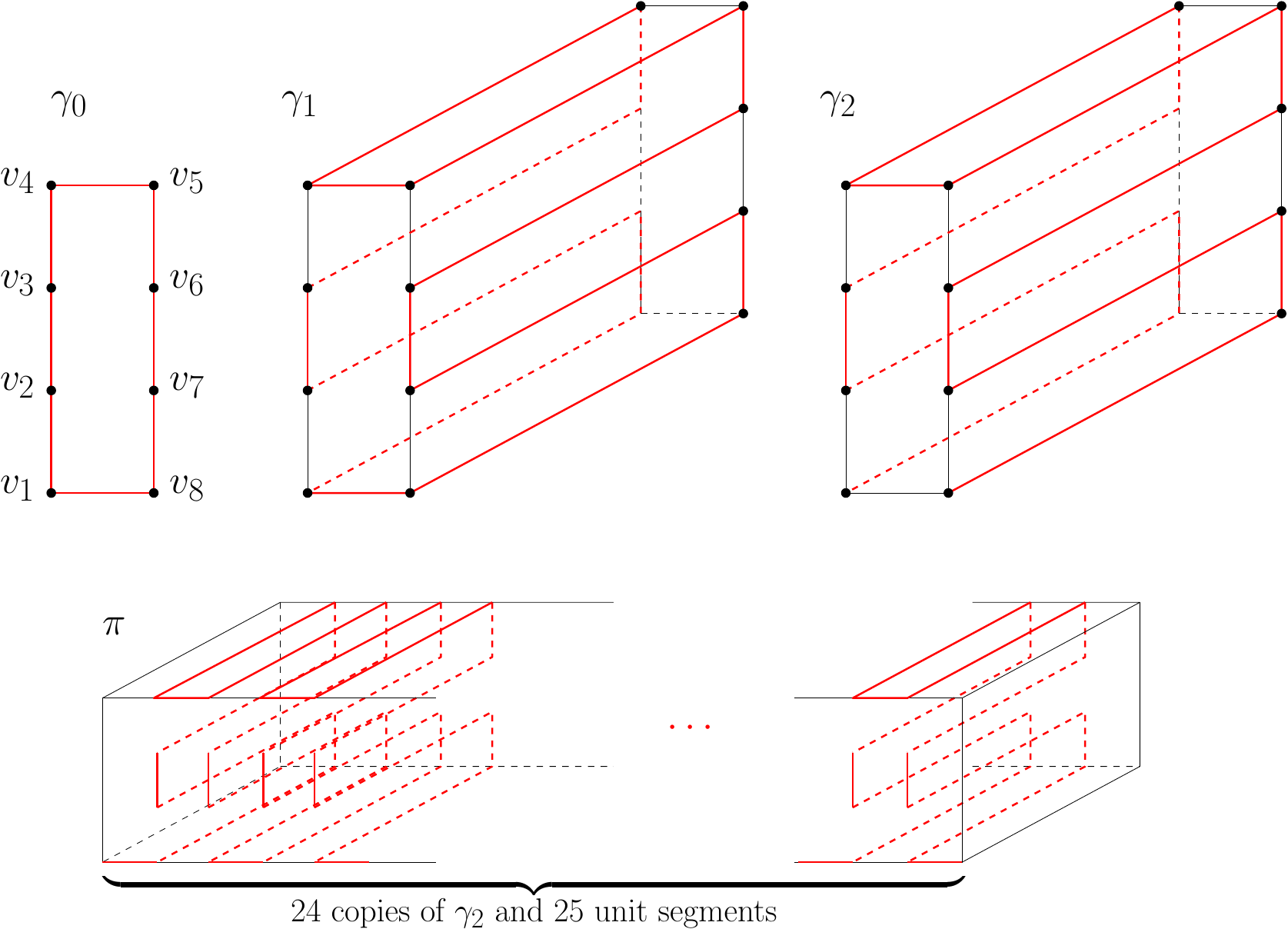}
\caption{The cycles $\gamma_0$ (top left), $\gamma_1$ (top middle),
and the chains $\gamma_2$ (top right), $\pi$ (bottom) for $d=3$ and $r=4$.
The cycles and chains are in red, their bounding boxes are outlined in black.}\label{fig:pi}
\end{figure}

\begin{lemma}\label{lem:pi}
The chain $\pi$ is a $c'$-chain for $c'=8+2r\sqrt{d-1}$.
Furthermore, if the points $q_1$, $q_2$, and $q_3$ are contained in
active edges,
in this order along $\pi$ and not all in the same edge, then
\[\frac{|q_1q_2|+|q_2q_3|}{|q_1q_3|} \leq 8+2r\sqrt{d-1}.\]
\end{lemma}
\begin{proof}
We extend $\pi$ to a chain
$\pi'$ by attaching a parallel copy of $\gamma_2$ to each end of $\pi$.
We prove the lemma for $\pi'$. Then, the lemma also follows for
$\pi$, as $\pi$ is a subchain of $\pi'$.
Write $\pi'=(p_1,\ldots , p_n)$.
Since $p_i$, $p_j$, and $p_k$ are endpoints of active edges,
for any choice of $1\leq i<j<k\leq n$, the second claim in the lemma
implies that $\pi'$ is a $c'$-chain.

We give an upper bound for the ratio $(|q_1q_2|+|q_2q_3|)/|q_1q_3|$.
Recall that all the active edges in $\pi'$ come from the $3r^2/2+2$
translated copies of the chain $\gamma_2$;
and that $\gamma_2$ has vertices in an axis-aligned bounding box
$B=[0,1]\times [0,r-1]^{d-2}\times [0,a]$.
Denote by $B_0, B_1,\ldots, B_{3r^2/2}, B_{3r^2/2+1}$ the minimum
axis-aligned bounding boxes of the $3r^2/2+2$ translates of
$\gamma_2$ in $\pi'$.
Suppose that $q_1$, $q_2$, and $q_3$ are in $B_{i_1}$, $B_{i_2}$, and $B_{i_3}$,
respectively. By assumption, $i_1\leq i_2\leq i_3$.

If $i_1=i_3$, then $q_1$, $q_2$, and $q_3$ are in $B_{i_1}$.
Since $q_1$ and $q_3$ are not on the same active edge,
and since $\gamma_0$ has integer coordinates, we have $|q_1q_3|\geq 1$. Consequently,
\begin{align*}
\frac{|q_1q_2|+|q_2q_3|}{|q_1q_3|}
&\leq \frac{2\cdot\diam\left(B_{i_1}\right)}{1}\\
&\leq 2\sqrt{1^2+(d-2)(r-1)^2+a^2}\\
&= 2\sqrt{1+(d-2)(r-1)^2+(r+1/(3r))^2}\\
&\leq 2\sqrt{2+(d-1)r^2}\\
&< 2\sqrt{2}+2r\sqrt{d-1}.
\end{align*}

Otherwise $i_1<i_3$, and the first coordinates of
$q_1$ and $q_3$ differ by at least $2(i_3-i_1)-1\geq i_3-i_1$, hence $|q_1q_3|\geq i_3-i_1$.
In this case,
\begin{align*}
\frac{|q_1q_2|+|q_2q_3|}{|q_1q_3|}
&\leq \frac{2\cdot\diam(B_{i_1}\cup B_{i_3})}{i_3-i_1}\\
&\leq \frac{2\cdot\sqrt{(2(i_3-i_1)+1)^2+(d-2)(r-1)^2+a^2}}{i_3-i_1}\\
&\leq \frac{4(i_3-i_1)+4+2r\sqrt{d-1}}{i_3-i_1}\\
&\leq 8+2r\sqrt{d-1},
\end{align*}
as claimed.
This completes the proof of Lemma~\ref{lem:pi}.
\end{proof}

\begin{figure}[htbp]
\centering
\includegraphics[width=0.75\textwidth]{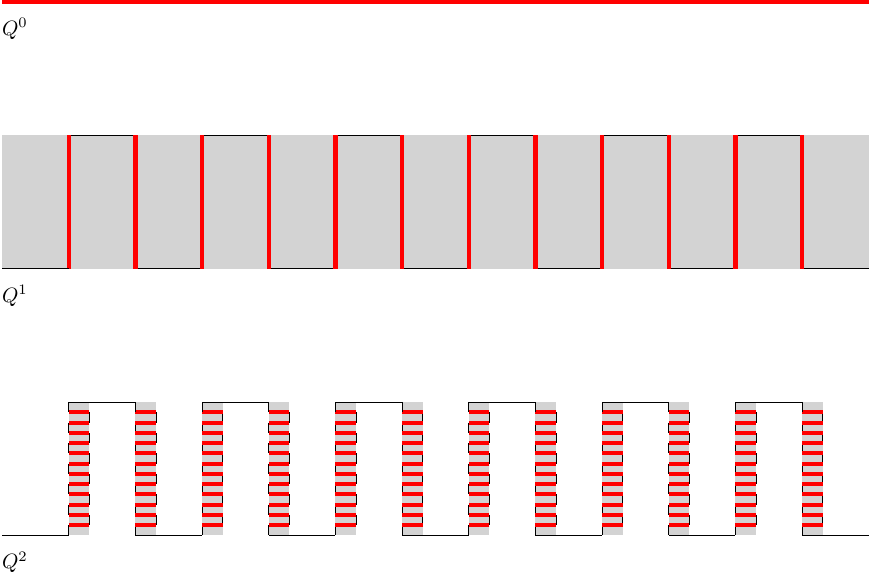}
\caption{The chains $Q^0$ (top), $Q^1$ (middle), and $Q_2$ (bottom) for $d=r=2$.
The active edges are highlighted by red bold lines.
The bounding box $B$ of $Q^1$ and bounding boxes $B'$ of homothetic copies of
$Q^1$ in $Q^2$ are shaded.}\label{fig:qqq}
\end{figure}

\begin{figure}[htbp]
\centering
\includegraphics[width=\textwidth]{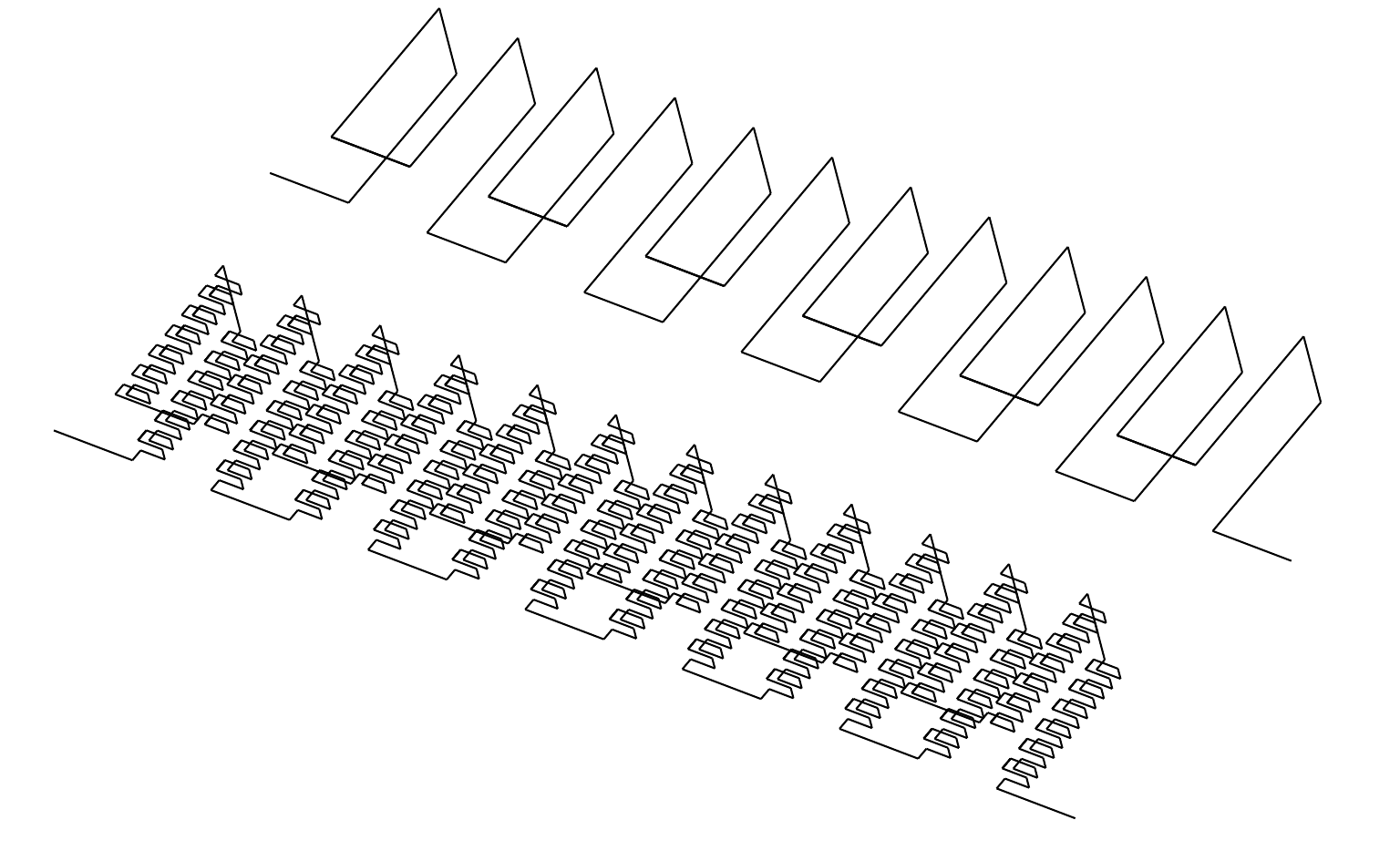}
\caption{The chains $Q^1$ (top) and $Q^2$ (bottom) for $d=3$ and $r=2$.}\label{fig:r2d3}
\end{figure}

Now the axis-parallel chains $Q^m$ can be defined recursively (see Fig.~\ref{fig:qqq} for an illustration).
Let $Q^0$ be a line segment of length $3r^2+1$, parallel to the $x_1$-axis, labeled active.
Let $Q^1$ be $\pi$ and let $B=B(\pi)$ be its minimum axis-parallel bounding box.
Recall that $B=\left[0,3r^2+1\right]\times [0,r-1]^{d-2}\times [0,a]$.

We maintain the invariant that each chain $Q^m$ ($m\in \mathbb{N}$) is contained in $B$.
In order to do this, let $B'$ be a hyperrectangle obtained from $B$ by a rotation of $90$ degrees
in the $\langle\mathbf{e}_1,\mathbf{e}_d\rangle$ plane, and scaling by a factor of $a/(3r^2+1)=1/(3r)$;
\ie, $B'=[0,a/(3r)]\times [0, (r-1)/(3r)]^{d-2}\times [0, a]$.
In particular, the longest edges of $B'$ are parallel to the active edges in $B$, and they all have length $a$.
Place a translate of $B'$ along each active edge in $Q^1$ such that all such translates
are contained in $B$. Note that the distance between any two translates is at least $1-2a/(3r)=1/3-2/(9r^2)\geq 5/18$.

For all $m\geq 1$, we construct $Q^{m+1}$ by replacing the active edges of $Q^1$ with
a scaled (and rotated) copy of $Q^m$ in each translate of $B'$; and we let the active edges of $Q^{m+1}$
be the active edges in these new copies of $Q^m$.

Instead of keeping track of the total length of $Q^m$,
we analyze the total length of the active edges of $Q^m$.
In each iteration, the number of active edges increases by a factor of
$3r^d$ and the length of an active edge decreases by a factor of $a/(3r^2+1)=1/(3r)$.
Overall the total length of active edges increases by a factor of $r^{d-1}$.
It follows that for all $m\in \mathbb{N}$, the chain $Q^m$ has $3^m r^{dm}$ active edges,
and their total length is $(3r^2+1)\cdot r^{(d-1)m}$.
Thus, we have
\begin{equation}
\label{equ:qmlength}
|Q^m| \geq (3r^2+1)\cdot r^{(d-1)m},
\end{equation}
for $m \in \mathbb{N}$.
Next we estimate the number of vertices in $Q^m$.
Recall that the recursive construction replaces each active edge
with $3r^d$ active edges and $3r^d+1$ inactive edges (which are never replaced).
Consequently, for $m\geq 1$, the number of inactive edges in $Q^m$
is $(3r^d+1)\sum_{i=0}^{m-1}3^i r^{di}$,
and the total number of vertices is
\[
n_m =1+ 3^m r^{dm}+ (3r^d+1)\sum_{i=0}^{m-1}3^i r^{di}= 1+ 3^m r^{dm}+ (3r^d+1)\frac{3^m r^{dm}-1}{3r^d-1}.
\]
Note that
\begin{equation}\label{eq:nm}
3^m r^{dm} < n_m \leq 3\cdot 3^m r^{dm}.
\end{equation}
Since the distance between the two endpoints of $Q^m$ remains $3r^2+1$,
we can use~\eqref{equ:qmlength} and the upper bound in \eqref{eq:nm}
to obtain %
\begin{equation} \label{equ:stretch}
\frac{|Q^m|}{3r^2+1}
\geq r^{(d-1)m}
\geq \left(\frac{n_m}{3^{m+1}}\right)^{\frac{d-1}{d}}.
\end{equation}
Now,
\eqref{eq:r} implies that $r = \beta \cdot 3^{(1-\eps)/(d\eps)}$, for
a constant $\beta > 1$.
Thus, using the lower bound in \eqref{eq:nm}, we get that
\[
n_m^\eps > 3^{\eps m} r^{\eps d m}
  = 3^{\eps m} \left( \beta \cdot 3^{\frac{(1-\eps)}{\eps d}} \right)^{\eps d m}
  = \beta^{\eps d m} \cdot 3^m \geq 3^{m + 1},
\]
for sufficiently large $m$. Hence,
combining with \eqref{equ:stretch}, we can bound the stretch
factor from below as
\[ \frac{|Q^m|}{3r^2+1}
\geq n_m^{(1-\eps)\frac{d-1}{d}},
\]
for sufficiently large $m$.

It remains to show that $\mathcal{Q}_c=\{Q^m: m\in \NN\}$ is a family of $c$-chains,
where $c=\Omega(d)$. We proceed by induction on $m$. The claim is trivial
for $m=0$, and it follows from Lemma~\ref{lem:pi} for $m=1$.

Now, let $m\geq 2$. Write $Q^m=(p_1,\ldots , p_n)$, and let $1\leq i<j<k\leq n$.
We shall derive an upper bound for the ratio $(|p_ip_j|+|p_jp_k|)/|p_ip_k|$.
Recall that $Q^m$ is obtained by replacing each active edge of $Q^1 = \pi$ by a
scaled copy of $Q^{m-1}$.
If $p_i$ and $p_k$ are in the same copy of $Q^{m-1}$, then so is $p_j$ and
induction completes the proof.

Otherwise let $B_i'$, $B_j'$, and $B_k'$ be the bounding boxes of the copies of $Q^{m-1}$
that contain $p_i$, $p_j$, and $p_k$, respectively.
Let $a_i$, $a_j$, and $a_k$ be the active segments in $Q^1$ that are replaced
by $B_i'$, $B_j'$, and $B_k'$;
and let $q_i\in a_i$, $q_j\in a_j$, and $q_k\in a_k$ be the
orthogonal projections of $p_i$, $p_j$, and $p_k$ onto $a_i$, $a_j$, and $a_k$, respectively.
(If $i=1$, then let $q_i=p_1$; if $k=n$, then let $q_k=p_n$.
Since the proof of Lemma~\ref{lem:pi} works on the extended chain $\pi'$, it applies to
$q_i$, $q_j$, and $q_k$ regardless of this special condition.)

Since each projection happens within a hyperplane orthogonal to
the $x_d$-axis onto an active edge in a translated copy of
$[0,a/(3r)]\times [0, (r-1)/(3r)]^{d-2}\times [0, a]$,
we have that $|p_iq_i|$, $|p_jq_j|$, and $|p_kq_k|$ are each
bounded above by
\[
\sqrt{\frac{a^2}{(3r)^2}+(d-2)\frac{(r-1)^2}{(3r)^2}}\leq \frac{\sqrt{d-1}}{3}+\frac{1}{3r} \leq \frac{\sqrt{d-1}}{3}+\frac16.
\]
As there are at least two distinct active edges among $a_i$, $a_j$, and $a_k$
(and as the distance between $p_1$ or $p_n$ and any active edge in $\pi$
is at least $1$),
we have
\[
|q_iq_j|+|q_jq_k|\geq \max\{|q_iq_j|,|q_jq_k|\}\geq 1.
\]
Combining these two bounds with the triangle inequality, we get
\begin{align*}
|p_ip_j|+|p_jp_k|
&\leq (|p_iq_i|+|q_iq_j|+|q_jp_j|)+(|p_jq_j|+|q_jq_k|+|q_kp_k|)\\
&\leq |q_iq_j|+|q_jq_k|+\frac43\,\sqrt{d-1}+\frac23\\
&\leq \left(\frac53+\frac43\,\sqrt{d-1}\right)(|q_iq_j|+|q_jq_k|).
\end{align*}
On the other hand, we have $|p_ip_k|\geq \frac{5}{18} |q_iq_k|$,
as this lower bound holds for the projections  of the edges to
each coordinate axis. Now Lemma~\ref{lem:pi} yields
\begin{align*}
\frac{|p_ip_j|+|p_jp_k|}{|p_ip_k|}
&\leq \frac{5/3+4\sqrt{d-1}/3}{5/18} \cdot \frac{|q_iq_j|+|q_jq_k|}{|q_iq_k|}\\
&\leq (6+24\,\sqrt{d-1}/5)\cdot (8+2r\sqrt{d-1})\\
&= O(r(d-1)).
\end{align*}
This completes the proof of Theorem~\ref{thm:d-lower-bound}.
\end{proof}

\section{Algorithm for Recognizing \texorpdfstring{$c$-Chains}{c-Chains}}
\label{sec:algo}

In this section, we design a randomized Las Vegas algorithm to recognize $c$-chains
in $d$-dimensional Euclidean space.
More precisely, given a polygonal chain $P=(p_1,\ldots ,p_n)$ in $\RR^d$, and a parameter $c\geq 1$,
the algorithm decides whether $P$ is a $c$-chain,
in $O\left(n^{3-1/d}\ {\rm polylog}\ n\right)$ expected time.
By definition, $P=(p_1,\dots, p_n)$ is a $c$-chain if
$|p_ip_j| + |p_jp_k|\leq c\ |p_ip_k|$ for all $1\leq i<j<k\leq n$;
equivalently, $p_j$ lies in the ellipsoid of major axis $c$
with foci $p_i$ and $p_k$. Consequently, it suffices to test,
for every pair $1\leq i<k\leq n$, whether the ellipsoid
of major axis $c|p_ip_k|$ with foci $p_i$ and $p_k$ contains
$p_j$, for all $j$, $i<j<k$. For this, we can apply recent
results from geometric range searching.

\begin{theorem}\label{thm:alg}
  For every integer $d\geq 2$,
  there are randomized algorithms that can decide, for a polygonal chain $P=(p_1,\dots , p_n)$ in $\RR^d$
  and a threshold $c>1$, whether $P$ is a $c$-chain in
\later{
  \begin{itemize}
  \item}
  $O\left(n^{3-1/d}\ {\rm polylog}\ n\right)$ expected time and $O(n\log n)$ space.
\later{, or
  \item $O\left(n^{3-2/(2d+1)+\eps}\right)$ expected time and space for any given $\eps>0$.
  \end{itemize}
    }
\end{theorem}

Agarwal, Matou\v{s}ek and Sharir~\cite[Theorem~1.4]{AMS13} constructed, for a set $S$ of $n$ points
in $\RR^d$, a data structure that can answer semi-algebraic range searching queries;
in particular, it can report the number of points in $S$ that are contained in a query ellipsoid.
Specifically, they showed that, for every $d\geq 2$ and $\eps>0$, there is a constant $B$
and a data structure with $O(n)$ space, $O\left(n^{1+\eps}\right)$ expected preprocessing time,
and $O\left(n^{1-1/d}\log^B n\right)$ query time. The construction was later simplified by
Matou\v{s}ek and Pat\'akov\'a~\cite{MP15}.
Using this data structure, we can quickly decide whether a given polygonal chain is a $c$-chain.

\begin{proof}[Proof of Theorem~\ref{thm:alg}]
  Subdivide the polygonal chain $P=(p_1,\dots , p_n)$ into
  two equal-sized subchains (to within $1$)
  $P_1=(p_1,\dots, p_{\lceil n/2\rceil})$ and $P_2=(p_{\lceil n/2\rceil},\dots, p_n)$;
  and recursively subdivide $P_1$ and $P_2$ until reaching 1-vertex chains.
  Denote by $T$ the recursion tree. Then, $T$ is a binary tree of depth $\lceil \log n\rceil$.
  There are at most $2^i$ nodes at level $i$; the nodes at level $i$ correspond to edge-disjoint subchains of $P$,
  each of which has at most $n/2^i$ edges. Let $W_i$ be the set of subchains on level $i$ of $T$;
  and let $W=\bigcup_{i\geq 0}W_i$. We have $|W|\leq 2n$.

  For each polygonal chain $Q\in W$, construct an ellipsoid range searching data structure
  $\DS(Q)$ described above~\cite{AMS13} for the vertices of $Q$, with a suitable parameter $\eps>0$.
  Their overall expected preprocessing time is
\[
\sum_{i=0}^{\lceil \log n\rceil} 2^i\cdot  O\left( \left(\frac{n}{2^i}\right)^{1+\eps} \right)
=O\left(n^{1+\eps}\sum_{i=0}^{\lceil \log n\rceil} \left(\frac{1}{2^i}\right)^{\eps}\right)
=O\left(n^{1+\eps}\right),
\]
and their space requirement is
$\sum_{i=0}^{\lceil \log n\rceil} 2^i\cdot  O\left(n/2^i\right)=O(n\log n)$.
The query time of each chain in $W_i$ is $O\left(\left(n/2^i\right)^{1-1/d}\ {\rm polylog}\ \left(n/2^i\right)\right)$.

For each pair of indices $1\leq i<k\leq n$, we do the following. Let $E_{i,k}$ denote the ellipsoid
of major axis $c|p_ip_k|$ with foci $p_i$ and $p_k$. The chain $(p_{i+1},\dots , p_{k-1})$ is subdivided
into $O(\log n)$ maximal subchains in $W$, using at most two subchains from each
set $W_i$, $i=0,\dots ,\lceil \log n\rceil$. For each of these subchains $Q\in W$,
query the data structure $\DS(Q)$ with the ellipsoid $E_{i,k}$. If all queries are positive
(\ie, the count returned is $|Q|$ in \emph{all} queries), then $P$ is a $c$-chain;
otherwise there exists $j$, $i<j<k$, such that $p_j\notin E_{i,k}$,
hence $|p_ip_j|+|p_jp_k|>c|p_ip_k|$, witnessing that $P$ is not a $c$-chain.

The query time over all pairs $1\leq i<k\leq n$ is bounded above by
\begin{align*}
\binom{n}{2} \sum_{i=0}^{\lceil \log n\rceil}
2\cdot
O\left(\left(\frac{n}{2^i}\right)^{1-1/d}\ {\rm polylog}\ \left(\frac{n}{2^i}\right)\right)
&= \binom{n}{2}\cdot O\left(n^{1-1/d}\ {\rm polylog}\  n\right)\\
&=O\left(n^{3-1/d}\ {\rm polylog}\ n\right).
\end{align*}
This subsumes the expected time needed for constructing the structures $\DS(Q)$, for all $Q\in W$.
So the overall running time of the algorithm is $O\left(n^{3-1/d}\ {\rm polylog}\ n\right)$, as claimed.
\later{
\smallskip
Alternatively, we can use the ellipsoid range counting data structure with logarithmic query time, described above~\cite{AgarwalAEZ19}. In order to keep the space and preprocessing time under control, we consider only
subchains of length up to $O(n^\alpha)$, for a suitable constant $\alpha\in [0,1]$.
These subchains correspond to the bottom $\lceil \alpha\log n\rceil+1$ levels of the recursion tree $T$.
Specifically, let $W'=\bigcup_{i=\lceil \log n\rceil -\lceil \alpha\log n\rceil}^{\lceil \log n\rceil}W_i$.

For each polygonal chain $Q\in W'$, construct an ellipsoid range searching data structure
$\DS(Q)$ in~\cite{AgarwalAEZ19} for the vertices of $Q$, with a suitable parameter $\eps>0$.
  Their overall size and expected preprocessing time is
\begin{align*}
\sum_{i=0}^{\lceil \alpha\log n\rceil} n^{1-\alpha} 2^i\cdot  O\left( \left(\frac{n^\alpha}{2^i}\right)^{2d+1+\eps} \right)
&=O\left(n^{1-\alpha}\cdot n^{\alpha(2d+1+\eps)}\sum_{i=0}^{\lceil \alpha\log n\rceil} \left(\frac{1}{2^i}\right)^{2d+\eps}\right)\\
&=O\left(n^{1+\alpha(2d+\eps)}\right),
\end{align*}
and the query time of each $DS(Q)$, $Q\in W'$, is bounded by $O\left(\log \left(n^\alpha\right)\right) \leq  O(\log n)$.

For each pair of indices $1\leq i<k\leq n$, we do the following. Let $E_{i,k}$ denote the ellipsoid
of major axis $c|p_ip_k|$ with foci $p_i$ and $p_k$. The chain $(p_{i+1},\dots , p_{k-1})$ is subdivided
into $O(\log n)$ maximal subchains in $W'$, using
$O(n^{1-\alpha})$ chains of length $\Theta(n^\alpha)$ from $W'$,
at most two subchains from all $\lceil a\log n\rceil$ lower levels
of the recursion tree.
For each of these subchains $Q\in W'$, query the data structure $\DS(Q)$ with the ellipsoid $E_{i,k}$.
The query time over all pairs $1\leq i<k\leq n$ is bounded above by
\[
\binom{n}{2} \cdot O(n^{1-\alpha}) \cdot O(\log n)
=O\left(n^{3-\alpha}\ \log n\right).
\]
By setting $\alpha=2/(2d+1)$, both the preprocessing time and the overall query time are bounded by $O\left(n^{3-\alpha+\eps}\right)=O\left(n^{3-2/(2d+1)+\eps}\right)$, as claimed.
}
\end{proof}

In the decision algorithm in the proof of Theorem~\ref{thm:alg},
only the construction of the data structures $\DS(Q)$, $Q\in W$,
uses randomization,  which is independent of the value of $c$.
The parameter $c$ is used for defining the ellipsoid $E_{i,k}$, and the queries
to the data structures; this part is deterministic. Hence, we can find the optimal value
of $c$ by Megiddo's parametric search~\cite{Meg83} in the second part of the algorithm.

Megiddo's technique reduces an optimization problem to a corresponding
decision problem at a polylogarithmic factor increase in the running time. An optimization problem
is amenable to this technique if the following three conditions are met~\cite{Salowe04}:
(1) the objective function is monotone in the given parameter;
(2) the decision problem can be solved by evaluating bounded-degree polynomials, and
(3) the decision problem admits an efficient parallel algorithm
(with polylogarithmic running time using a polynomial number of processors).
All three conditions hold in our case: The area of each ellipsoid with foci in $S$ monotonically
increases with $c$; the data structure of~\cite{MP15} answers ellipsoid range counting queries
by evaluating polynomials of bounded degree; and the $\binom{n}{2}$ queries can be performed in parallel.
Alternatively, Chan's randomized optimization technique~\cite{Chan99} is also applicable. Both techniques
yield the following result.

\begin{corollary}
\label{cor:alg}
There are randomized algorithms that can find, for a polygonal chain $P=(p_1,\dots , p_n)$ in $\RR^d$,
the minimum $c\geq 1$ for which $P$ is a $c$-chain
in $O\left(n^{3-1/d}\ {\rm polylog}\ n\right)$ expected time and $O(n\log n)$ space.
\end{corollary}

We note that, for $c=1$, the test takes $O(n)$ time: it suffices
to check whether points $p_3,\dots, p_n$ lie on the
line spanned by $p_1p_2$, in that order.

\subparagraph*{Remark.}
Recently, Agarwal et al.~\cite[Theorem~13]{AgarwalAEZ19} designed a data structure for semi-algebraic range searching queries that supports $O(\log n)$ query time, at the expense of higher space and preprocessing time.
The size and preprocessing time depend on the number of free parameters that describe the semi-algebraic set.
An ellipsoid in $\RR^d$ is defined by $2d+1$ parameters: the coordinates of its foci and the length of its major axis.
Specifically, they showed that, for every $d\geq 2$ and $\eps>0$, there is a data structure with $O(n^{2d+1+\eps})$ space and $O(n^{2d+1+\eps})$ expected preprocessing time that can report the number of points in $S$ contained in a query ellipsoid in $O(\log n)$ time.
This data structure allows for a tradeoff between preprocessing time and overall query time in the algorithm above.
However the resulting tradeoff does not seem to yield an improvement over the expected running time in
Theorem~\ref{thm:alg} for any $d\geq 2$.

\section{Conclusion}\label{sec:conclusion}
We conclude with some remarks and open problems.
\smallskip

\begin{enumerate}
\item
The lower bound construction in the plane can be slightly improved as follows.
For $m\geq 1$, let $P^m_* = g_2(P^m)\cup g_3(P^m)$, see \textsc{Fig.}~\ref{fig:refined-chain}\,(right).
Observe that $P^m_*$ is a $c$-chain with $n=4^m/2+1$ vertices and stretch factor
\[\sqrt{c(c-2)/8}(n-1)^{\frac{1+\log(c-2)-\log c}{2}}.\]
Since $\sqrt{c(c-2)/8}\geq 1$ for $c\geq 4$, this improves the result of Theorem~\ref{thm:lower-bound}
by a constant factor. Since this construction does not improve the exponent, and the analysis would be
longer (requiring a case analysis without new insights), we omit the details.
\begin{figure}[htbp!]
\centering
\hspace*{15mm}
\begin{tikzpicture}
\path (0,0) pic[scale=0.7] {p4};
\end{tikzpicture}
\hfill
\begin{tikzpicture}[scale=0.7]
\draw (4.545, 0.000) -- (4.588, 0.425) -- (4.170, 0.510) -- (4.597,
0.510) -- (4.639, 0.934) -- (4.221, 1.019) -- (4.095, 0.612) --
(4.137, 1.036) -- (3.719, 1.121) -- (4.146, 1.121) -- (4.189, 1.546)
-- (4.231, 1.121) -- (4.658, 1.121) -- (4.701, 1.546) -- (4.282,
1.631) -- (4.709, 1.631) -- (4.752, 2.056) -- (4.334, 2.141) --
(4.207, 1.733) -- (4.250, 2.158) -- (3.832, 2.243) -- (3.705, 1.835)
-- (4.098, 1.668) -- (3.680, 1.753) -- (3.554, 1.346) -- (3.596,
1.770) -- (3.178, 1.855) -- (3.605, 1.855) -- (3.648, 2.280) --
(3.229, 2.365) -- (3.103, 1.957) -- (3.146, 2.382) -- (2.727, 2.467)
-- (3.154, 2.467) -- (3.197, 2.892) -- (3.240, 2.467) -- (3.666,
2.467) -- (3.709, 2.892) -- (3.291, 2.977) -- (3.718, 2.977) --
(3.760, 3.401) -- (3.803, 2.977) -- (4.230, 2.977) -- (3.812, 2.892)
-- (3.854, 2.467) -- (4.281, 2.467) -- (4.324, 2.892) -- (4.367,
2.467) -- (4.793, 2.467) -- (4.836, 2.892) -- (4.418, 2.977) --
(4.845, 2.977) -- (4.887, 3.401) -- (4.469, 3.486) -- (4.343, 3.079)
-- (4.385, 3.503) -- (3.967, 3.588) -- (4.394, 3.588) -- (4.437,
4.013) -- (4.479, 3.588) -- (4.906, 3.588) -- (4.949, 4.013) --
(4.530, 4.098) -- (4.957, 4.098) -- (5.000, 4.523) -- (5.043, 4.098)
-- (5.470, 4.098) -- (5.051, 4.013) -- (5.094, 3.588) -- (5.521,
3.588) -- (5.563, 4.013) -- (5.606, 3.588) -- (6.033, 3.588) --
(5.615, 3.503) -- (5.657, 3.079) -- (5.531, 3.486) -- (5.113, 3.401)
-- (5.155, 2.977) -- (5.582, 2.977) -- (5.164, 2.892) -- (5.207,
2.467) -- (5.633, 2.467) -- (5.676, 2.892) -- (5.719, 2.467) --
(6.146, 2.467) -- (6.188, 2.892) -- (5.770, 2.977) -- (6.197, 2.977)
-- (6.240, 3.401) -- (6.282, 2.977) -- (6.709, 2.977) -- (6.291,
2.892) -- (6.334, 2.467) -- (6.760, 2.467) -- (6.803, 2.892) --
(6.846, 2.467) -- (7.273, 2.467) -- (6.854, 2.382) -- (6.897, 1.957)
-- (6.771, 2.365) -- (6.352, 2.280) -- (6.395, 1.855) -- (6.822,
1.855) -- (6.404, 1.770) -- (6.446, 1.346) -- (6.320, 1.753) --
(5.902, 1.668) -- (6.295, 1.835) -- (6.168, 2.243) -- (5.750, 2.158)
-- (5.793, 1.733) -- (5.666, 2.141) -- (5.248, 2.056) -- (5.291,
1.631) -- (5.718, 1.631) -- (5.299, 1.546) -- (5.342, 1.121) --
(5.769, 1.121) -- (5.811, 1.546) -- (5.854, 1.121) -- (6.281, 1.121)
-- (5.863, 1.036) -- (5.905, 0.612) -- (5.779, 1.019) -- (5.361,
0.934) -- (5.403, 0.510) -- (5.830, 0.510) -- (5.412, 0.425) --
(5.455, 0.000);
\end{tikzpicture}
\caption{The chains $P^4$ (left) and $P^4_*$ (right).}\label{fig:refined-chain}
\end{figure}
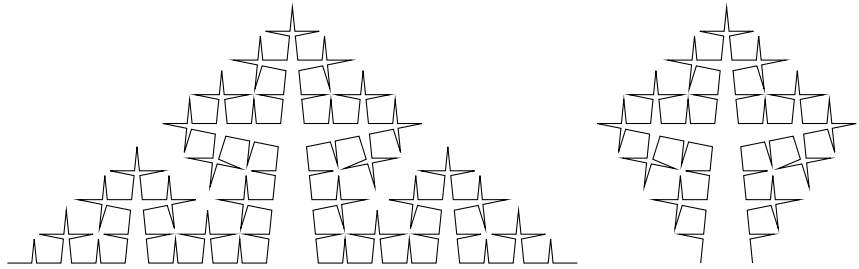

\item
The lower bound construction in the plane depends on a parameter
$c_*=(c-2)/2$. If $c$ were used instead,
the condition $c\geq 4$ in Theorem~\ref{thm:lower-bound} could be replaced
by $c\geq 1$, and the bound could be improved from
\[(n-1)^{\frac{1+\log(c-2)-\log c}{2}} \quad\text{to}\quad (n-1)^{\frac{1+\log c-\log(c+1)}{2}}.\]
Although we were unable to prove that the resulting $P^m$'s, $m\in\NN$,
are $c$-chains, a computer program has verified that the first
few generations of them are indeed $c$-chains.

\item
The upper bounds in Theorems~\ref{thm:logc}--\ref{thm:1/2}
(and their generalizations to higher dimensions, \eg, Theorem~\ref{thm:dspace})
are valid regardless of whether the chain is crossing or not.
On the other hand, the lower bounds in Theorem~\ref{thm:lower-bound}
and Theorem~\ref{thm:d-lower-bound} are given by noncrossing chains.
A natural question is whether sharper upper bounds hold if the chains
are required to be noncrossing.
Specifically, can the exponent of $n$ in the upper bound for $\RR^d$ be reduced
to $\frac{d-1}{d}-\eps$, where $\eps>0$ depends on $c$?

\item
The running time of the algorithm in Theorem~\ref{thm:alg} is sub-cubic,
but super-quadratic.
Is this necessary, or is it possible to decide the $c$-chain property
in time $O(n^2)$ or better?
\end{enumerate}

\bibliographystyle{plainurl}
\bibliography{stretch}

\end{document}